\documentclass[aps,pra,twocolumn,showpacs,superscriptaddress, nofootinbib]{revtex4-2} 
\pdfoutput=1

\usepackage{graphicx}
\usepackage{titletoc}
\usepackage{physics}
\usepackage{enumitem}
\usepackage{bbold}
\usepackage{amsmath, amsthm, amssymb,mathrsfs,bm}

\newtheorem{definition}{Definition}
\newtheorem{remark}{Remark}
\newtheorem{theorem}{Theorem}

\usepackage{mathtools}

\usepackage[normalem]{ulem}
\usepackage{optidef}
\usepackage{algorithmic}
\usepackage{algorithm}

\usepackage{algorithm}

\newcommand{\be}{\begin{equation}}
\newcommand{\ee}{\end{equation}}
\newcommand{\cE}{\mathcal{E}}

\usepackage{dsfont}
\usepackage{comment}
\usepackage{vwcol}
\usepackage[colorlinks]{hyperref}
\usepackage[usenames,dvipsnames,table]{xcolor}

\usepackage[acronym]{glossaries}
\makeglossaries
\newacronym{qkd}{QKD}{quantum key distribution}
\newacronym{di}{DI}{device-independent}
\newacronym{pm}{PM}{prepare-and-measure}
\newacronym{sdp}{SDP}{semidefinite programming}
\newacronym{POVM}{POVM}{Positive Operator Valued 
Measure}
\newacronym{usd}{USD}{Unambiguous state discrimination}
\newacronym{qber}{QBER}{Quantum Bit Error Rate}
\newcommand{\id}{\mathbb{1}}

\newcommand\cw{\circlearrowright}
\newcommand\acw{\circlearrowleft}

\usepackage{float}
\usepackage[caption=false]{subfig}

\usepackage{tikz}
\usetikzlibrary{
  arrows.meta,
  automata
}
\tikzset{
    -Latex,auto,node distance =1 cm and 1 cm,semithick,
    state/.style ={ellipse, draw, minimum width = 0.7 cm},
    point/.style = {circle, draw, inner sep=0.04cm,fill,node contents={}},
    bidirected/.style={Latex-Latex,dashed},
    el/.style = {inner sep=2pt, align=left, sloped}
}

\usepackage{multirow}
\usepackage{diagbox}

\usepackage{wrapfig}

\usepackage[normalem]{ulem}
\newcommand{\stkout}[1]{\ifmmode\text{\sout{\ensuremath{#1}}}\else\sout{#1}\fi}

\newcommand{\floor}[1]{\left\lfloor #1 \right\rfloor}

\begin{document}

\title{Topologically Robust Quantum Network Nonlocality}

\author{Sadra Boreiri}
\thanks{These authors contributed equally to this work.}
\affiliation{Department of Applied Physics University of Geneva, 1211 Geneva, Switzerland}
\author{ Tamas Krivachy}
\thanks{These authors contributed equally to this work.}
\affiliation{ICFO - Institut de Ciencies Fotoniques, The Barcelona Institute of Science and Technology, 08860 Castelldefels (Barcelona), Spain}
\affiliation{Atominstitut, Technische Universitat Wien, 1020 Vienna, Austria}
\author{Pavel Sekatski}
\thanks{These authors contributed equally to this work.}
\affiliation{Department of Applied Physics University of Geneva, 1211 Geneva, Switzerland}
\author{Antoine Girardin}
\affiliation{Department of Applied Physics University of Geneva, 1211 Geneva, Switzerland}

\author{Nicolas Brunner}
\affiliation{Department of Applied Physics University of Geneva, 1211 Geneva, Switzerland}

\begin{abstract}
We discuss quantum network Bell nonlocality in a setting where the network structure is not fully known. More concretely, an honest user may trust their local network topology, but not the structure of the rest of the network, involving distant (and potentially dishonest) parties. We demonstrate that quantum network nonlocality can still be demonstrated in such a setting, hence exhibiting topological robustness. Specifically, we present quantum distributions obtained from a simple network that cannot be reproduced by classical models, even when the latter are based on more powerful networks. In particular, we show that in a large ring network, the knowledge of only a small part of the network structure (involving only 2 or 3 neighbouring parties) is enough to guarantee nonlocality over the entire network. This shows that quantum network nonlocality can be extremely robust to changes in the network topology. Moreover, we demonstrate that applications of quantum nonlocality, such as the black-box certification of randomness and entanglement, are also possible in such a setting. 
\end{abstract}

\maketitle

Quantum networks play a key role in quantum information processing. A notable example is in quantum communication, where quantum networks represent the backbone infrastructure, allowing for the distribution and manipulation of entanglement over large scales \cite{Kimble2008,Simon2017,Wehner2018}. 
Such a network typically involves several distant parties (nodes), subsets of which are connected by different sources of entanglement. Quantum systems originating from different sources are then jointly processed at the nodes, allowing e.g. for the distribution of entanglement over the entire networks or the heralding of entanglement between a given pair of nodes.

So far the analysis of quantum networks typically focuses on the setting where the network structure is completely known and fixed. It is however interesting to investigate quantum networks in a setting where the network structure is only partially known, see e.g. \cite{Pappa2012,Murta2023}. For example, a party may only know/trust a small part of the network (in their local neighbourhood), but not the rest of the network. Beyond the conceptual interest, such a setting is clearly relevant in a practical context, for example in an adversarial scenario where a number of dishonest parties may want to take control over the whole network by collaborating. 

\begin{figure}[b!]
    \centering
     \includegraphics[width=0.98\columnwidth]{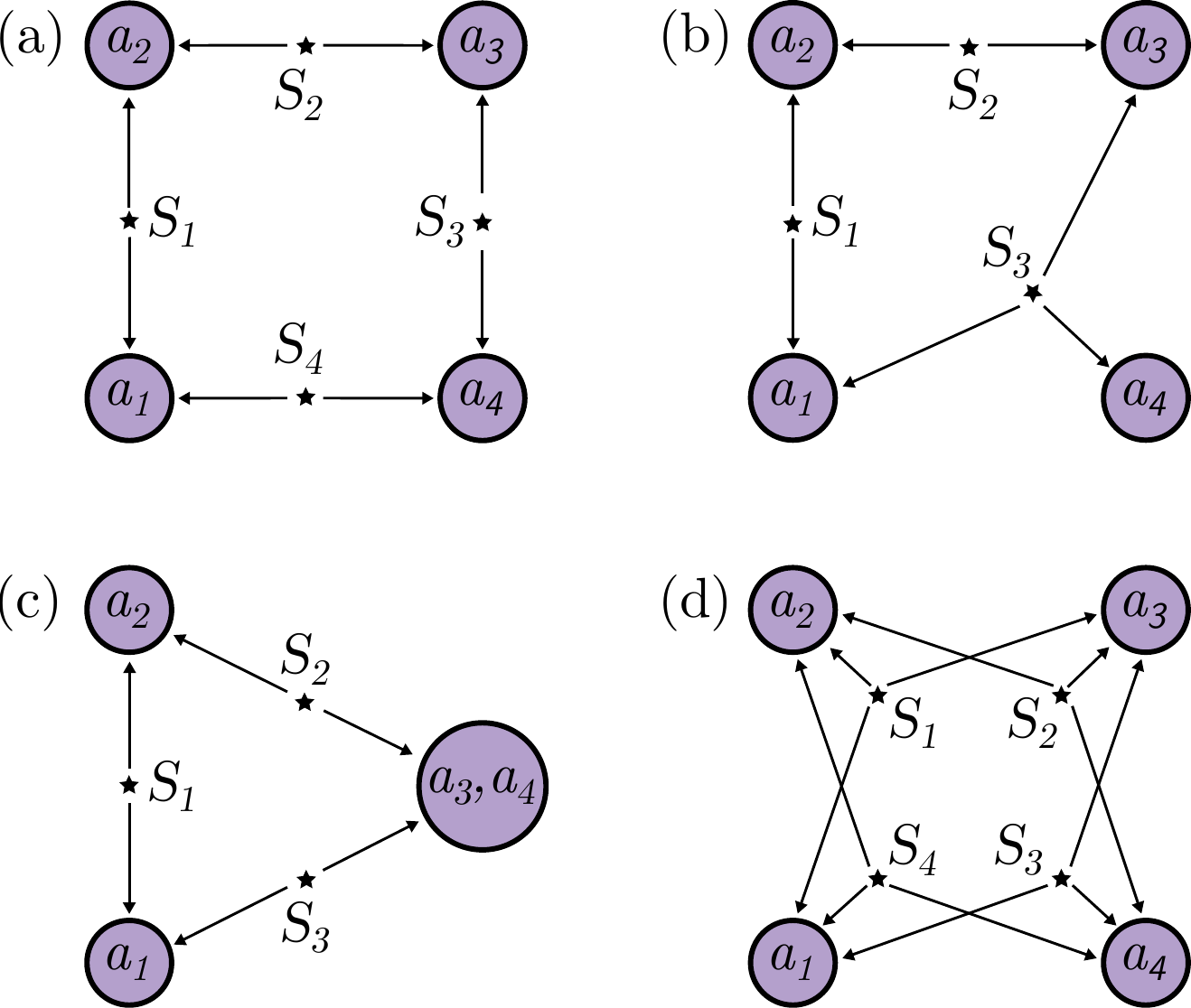}
     \caption{Different network configurations involving four parties (purple dots), connected by independent sources $S_1,S_2,S_3$ and $S_4$. Each party produces an outcome (denoted $a_1,a_2,a_3,a_4$). The strength of possible correlations, given by the joint distribution $P(a_1,a_2,a_3,a_4)$ depends on the network topology. In the main text, we present a quantum distribution on the square network (a), which remains nonlocal with respect to all four network structures, hence featuring topologically robust nonlocality.}
    \label{fig:sq_normal_extended}
\end{figure}

In this work, we discuss such a setting in the context of quantum network Bell nonlocality. The latter investigates quantum correlations in networks, in particular the advantage they offer over classical models, see e.g.  \cite{Branciard2010,Branciard_2012,Fritz_2012,Chaves2015,wolfe2019inflation,Aberg2020,wolfe2021inflation,Contreras2021,gisin2020constraints,Sekatski2023,Ligthart2023} and \cite{Tavakoli_Review} for a review. The main question we investigate here is the following: given only partial knowledge of the network structure, it is still possible to demonstrate the presence of quantum nonlocality? We will answer this question in the affirmative, and further show that black-box applications of quantum nonlocality, e.g. generating certified randomness, can be implemented in this setting. This shows that quantum Bell nonlocality in networks can be topologically robust.

In order to illustrate and formalize the problem, we start by considering a quantum distribution obtained on a simple ``square newtork'', with four parties in a ring configuration as in Fig.~\ref{fig:sq_normal_extended}(a). While this quantum distribution is nonlocal with respect to the (original) square network, we first show that it remains nonlocal even if we consider different network topologies as in Fig.~\ref{fig:sq_normal_extended}(b,c,d). That is, this quantum distribution (obtained in the square network) cannot be reproduced classically for any of these network structures, even though the latter enables local correlations that are much stronger compared to what is possible in the square network. Hence quantum nonlocality is here robust to changes in the network topology. In the second part of the paper, we demonstrate a much more dramatic instance of this phenomenon. We consider a large ring network with $N$ parties, and show that it is still possible to certify quantum nonlocality over the entire network while trusting only a very small part of the network topology; we present two constructions where we need to trust only the local network structure of two or three neighbouring parties. This shows that quantum nonlocality can be extremely robust to changes in the network topology.  

Furthermore, we show that applications of quantum nonlocality are possible in this setting where the network structure is only partially known. In particular, we demonstrate the presence of certified randomness in the local output of a given party, as well as the certification of entanglement, both at the level of states and measurements, which enables the certification of genuine quantum network nonlocality \cite{Supic2022}. At the technical level, a notable aspect of our work is to develop a systematic approach to self-test quantum distributions in
networks \cite{Sekatski2023}, leading to stronger bounds on randomness, which are also of independent interest. Finally, we conclude with a number of open questions.

\section{Problem and illustrative example}

To start our analysis, let us consider a network involving four distinct parties, denoted $A_j$ (with $j \in \{1,2,3,4\}$). The parties share physical resources, distributed by a number of independent sources, and different topologies can be considered, as illustrated in Fig. 1. Upon receiving these resources and processing them jointly, each party then provides a classical output denoted $a_j$. The correlation between these outputs is captured by the joint probability distribution $P(a_1,a_2,a_3,a_4)$.

In general, the strength of these correlations will depend on two important features. First, the nature of the physical systems distributed by the sources---notably quantum systems can lead to stronger correlations than classically possible (for a given network structure), the effect of quantum nonlocality (see e.g. \cite{Tavakoli_Review}). Second, the structure of the network itself, i.e. how the sources and parties are connected. This second aspect is the focus of the present work. In particular, we want to investigate the robustness of quantum nonlocality with respect to different topologies of the network.

To formalize the problem let us first discuss correlations for classical models. Here each source distributes a (classical) random variable (denoted $\lambda_k$ for source $k$) to all the parties connected to it. Importantly each source, i.e. the variables $\lambda_k$,  are assumed to be independent from the others \cite{Branciard2010,Fritz_2012}. For example, considering the square network of Fig.~1(a), possible correlations take the form
\begin{align} \label{localsquare}
P_L^\square(a_1,a_2,a_3,a_4)=\mathds{E}[
&p(a_1|\lambda_4,\lambda_1) p(a_2|\lambda_1,\lambda_2) \\ \nonumber  
&p(a_3|\lambda_2,\lambda_3)
p(a_4|\lambda_3,\lambda_4) ].
\end{align}
Here and below the expected value $\mathds{E}[.]$ is taken with respect to the random variables $\lambda_k$, assumed to be independent of each other. Any distribution admitting a model of this form is termed local; if no such decomposition exists, the distribution is termed nonlocal (with respect to the square network). 

Other network structures can be considered, allowing for more complex sources as in  Fig.~\ref{fig:sq_normal_extended}(b,d), or where two parties merge (forming a single party producing the corresponding outputs), as in Fig.~\ref{fig:sq_normal_extended}(c). From Eq. \eqref{localsquare} it is straightforward to define local correlations for these networks (see Appendix~\ref{app:preliminary}).

In general, characterizing sets of correlations achievable in networks is challenging; due to the independence condition of the sources these sets are not convex, see e.g. \cite{Branciard_2012,Rosset2016,wolfe2019inflation,Wolfe2021}. In the present case, we can nevertheless notice a hierarchy between sets of correlations for the four networks in Fig. \ref{fig:sq_normal_extended}. First, the network (b) is \textit{stronger} than (a), in the sense that all the models compatible with the latter are also possible in the former: the source $S_3$ in (b) can be composed of the sources $S_3$ and $S_4$ in (a). With similar arguments, we see that both (c) and (d) are stronger than (b). Finally, there is no strict relation between networks (c) and (d). Their respective sets of correlations are incomparable.

Let us now move to quantum models, where sources distribute quantum states. Interestingly, even though we consider a setting where each party performs a fixed measurement, the resulting output distribution can still feature quantum nonlocality \cite{Fritz_2012,branciard2012bilocal,Renou_2019}. We start with an example of a quantum distribution that exhibits nonlocality that is robust to modifications of the network topology. Consider the square network in Fig. 1 (a), which we just argued leads to the weakest correlations among all the networks in Fig. 1. Following \cite{Renou_2019} we consider that each source distributes a two-qubit Bell state $\ket{\psi_+} = \frac{\ket{10}+\ket{01}}{\sqrt{2}}$. Hence, each party receives two qubits, coming from two independent sources. In turn, to produce a four-valued output $a_j\in\{0,1_0,1_1,2\}$ each party performs a measurement in the following basis
\begin{equation}
\label{eq:measurment}
  \begin{split}
      & \ket{\bar{0}} = \ket{00}, \quad \ket{\bar{1}_0} =  u\ket{01}+v\ket{10},  \\ & \ket{\bar{1}_1} = v\ket{01}-u\ket{10},\quad \ket{\bar{2}} = \ket{11},
  \end{split}
\end{equation}
such that $u^2 + v^2 =1$ ($u,v \in \mathbb{R}$).  The resulting quantum distribution is denoted by $P_Q^\square(a_1,a_2,a_3,a_4) $. This distribution can exhibit topologically robust network nonlocality, as formalized in the following result.\\

\noindent{\bf Result 1.} The quantum distribution $P_Q^\square(a_1,a_2,a_3,a_4) $, for the parameter range $u_{\mathrm{min}}< u<1$ with $u_{\mathrm{min}} \approx 0.841$, is provably nonlocal with respect to the networks in Fig. 1, for configurations (a), (b) and (c). For the network (d), numerical results indicate that the quantum distribution is also nonlocal.\\

\noindent\emph{Sketch of proof.}  The full proof is given in Appendix \ref{app:A}. First note that the nonlocality of the distribution with respect to the square network (a) is proven in \cite{Renou_2019,Renou2022}. 
The key idea of the proof goes as follows. When the outcomes $a_j \in\{ 0,1_0,1_1,2\}$
are coarse-grained into $\tilde{a}_j \in \{ 0,1,2\}$ (by merging 
$ \{ 1_0, 1_1 \}\mapsto 1$), the resulting quantum distribution $P_Q^\square(\tilde a_1, \tilde a_2, \tilde a_3,\tilde a_4)$ satisfies the token counting (TC) property. That is, the condition $\sum_j \tilde{a}_j = 4$ holds deterministically. While such a coarse-grained distribution can be achieved via a local model on the square network, the latter turns out to be essentially unique (up to irrelevant relabellings)~\cite{Renou2022}. In turn, this property of ``rigidity'' leads to strong constraints, from which one can show that the original (fine-grained) quantum distribution cannot be reproduced by a local model.  

We prove the nonlocality of the quantum distribution $P_Q^\square$ with respect to the networks (b) and (c) using similar techniques. In fact, the rigidity property of TC distributions extends to any network as long as no pairs of parties are connected by more than one source \cite{Renou2022}. Since our fourth network (d) does not satisfy this property (as here any pair of parties is connected by two sources), we resort here to numerical techniques to investigate nonlocality. Specifically, we use the generative neural network algorithm developed in Ref. \cite{krivachy_neural_2020}, and obtain significant evidence that the quantum distribution $P_Q^\square$ cannot be achieved via a local model.

The quantum distribution $P_Q^\square$ thus represents an illustrative example of quantum network nonlocality that is topologically robust. In fact, this result can even be strengthened in the following way. We can prove that not only nonlocality, but also certain features of the quantum model can be robustly certified even though the network structure is only partially known. \\

\noindent{\bf Result 1'.} For the quantum distribution $P_Q^\square(a_1,a_2,a_3,a_4) $, for the parameter range $u_{\mathrm{min}}< u<1$ with $u_{\mathrm{min}} \approx 0.841$, and the networks (a),(b) and (c) in Fig.~1, we prove the following properties: \textit{(i)} the randomness of outcome $a_1$ is lower bounded, \textit{(ii)} the entanglement distributed by the sources is lower bounded, and \textit{(iii)} the measurements of parties connected to two sources must be entangled.\\

\noindent\emph{Sketch of the proof.}  The full proof and exact lower bounds are given in Appendix~\ref{app:A}. Since the network $(c)$ is the strongest of the three, it is enough to exhibit the properties \textit{(i-iii)} for this network. Similarly to the proof of Result 1, the key concept here is again rigidity. However, this time we use a notion of rigidity with respect to quantum models, following the techniques developed in Ref. \cite{Sekatski2023}. Note that each of the points $(i-iii)$ implies that the distribution is nonlocal. $\square$

Finally, we note that Result 1 and Result 1' also apply to any network that can be embedded in network (c), for example, the fully connected graph where each pair of parties are connected by a source. Furthermore, for the proof of Result 1', we derive new bounds on the output randomness applicable to a broad class of networks,  improving the previously known bound \cite{Sekatski2023} for the triangle network by a factor of four (see Appendix B.3.c and Fig. 4).


\section{Topologically robust nonlocality for large networks}

\begin{figure*}[!ht]
    \centering
    \includegraphics[width=\textwidth]{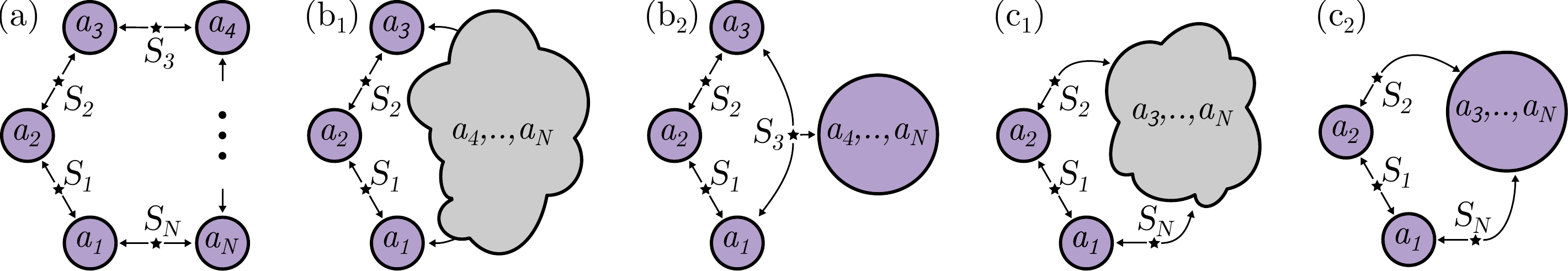}
    \caption{We consider a quantum distribution $P_Q^{\bigcirc}(a_1,\dots,a_N)$ obtained in the $N$-party ring network in (a). We show that the nonlocality of this distribution can still be demonstrated in a setting where only a small part of the network topology is trusted, and we consider two configurations. The first is shown in (b$_1$), where we trust only the topology for three neighbouring parties ($A_1$ to $A_3$). To prove nonlocality with respect to any such network, we consider the strongest possible network compatible with (b$_1$), where all parties $A_4$ to $A_N$ merge, as shown in (b$_2$). Second, we consider configuration in (c$_1$) where we trust only the local topology of two neighbouring parties ($A_1$ and $A_2$). Again the nonlocality of the quantum distribution $P_Q^{\bigcirc}(a_1,\dots,a_N)$ is demonstrated for any network of the form (c$_1$). This is proven by considering the network in (c$_2$), the strongest one compatible with (c$_1$).
    }
    \label{fig:general}
\end{figure*}

Having set and illustrated the concepts of topologically robust nonlocality, let us now move to larger quantum networks involving $N$ parties. We consider a quantum distribution $P_Q^\bigcirc$ on a ring network, as in Fig.~\ref{fig:general}(a), and show that its nonlocality features very strong topological robustness. More specifically, we present two slightly different scenarios. First we consider the configuration in Fig.~\ref{fig:general}(b$_1$) where we trust the network structure for only three neighbouring parties; $A_1$, $A_2$ and $A_3$ and two connecting sources $S_1$ and $S_2$. We prove that the quantum distribution $P_Q^\bigcirc$ is nonlocal for any possible network structure for the $N-3$ remaining parties ($A_4$ to $A_N$), and that randomness and entanglement can be certified. Second, we discuss the configuration of Fig.~\ref{fig:general}(c$_1$), where we trust only the network structure for two parties ($A_1$ and $A_2$ and three connecting sources $S_1$, $S_2$ and $S_3$), and prove nonlocality of the distribution $P_Q^\bigcirc$.

Formally, we consider a ring network with $N$ parties ($A_1$ to $A_N$) such that each pair of parties $A_i$ and $A_{i+1}$ connected by a bipartite source $S_i$ (with $A_{N+1}=A_1$). The quantum model we consider is similar to the one discussed above. That is, each source distributes a Bell state $\ket{\psi^+}$ and each party performs the two-qubit measurement in Eq.~\eqref{eq:measurment} to produce a four-valued outcomes $a_j\in\{\bar 0, \bar 1_0,\bar 1_1,\bar 2\}$. The resulting distribution $P_Q^\bigcirc (a_1\dots,a_N)$ is straightforward to compute. We now show that its nonlocal properties are topologically robust.

First, let us consider the setting depicted in Fig.~\ref{fig:general}(b$_1$). That is, we trust the local topology of the network involving the three parties $A_1$, $A_2$ and $A_3$ (connected in a chain), while the rest of the network topology is unknown. For example, one could imagine that all other parties $A_4$ to $A_N$ are in fact collaborating as in Fig.~\ref{fig:general}(b$_2$). Even though we trust here only a very small part of the network, we can still prove nonlocality, as well as bounds on certified entanglement and randomness. \\

\noindent{\bf Result 2.} Consider any network structure as in Fig.\ref{fig:general}(b$_1$), i.e. fixing only the local topology of parties $A_1$, $A_2$ and $A_3$. For $N\geq 4$ and considering the parameter regime $u_{\text{min}}< u<1$ with $u_{\text{min}}=0.841$, 
the quantum distribution $P_Q^{\bigcirc}(a_1,\dots,a_N)$ is nonlocal and has the following properties: \textit{(i)} the randomness of $a_2$ is lower bounded 
, \textit{(ii)} the entanglement of $S_1$ and $S_2$ is lower bounded, \textit{(iii)} the measurements of $A_1$, $A_2$ and $A_3$ are entangled.\\

\noindent \emph{Sketch of the proof.} The full proof and detailed lower bounds are given in Appendix~\ref{app: result 2 and 3}, here we present a sketch of the nonlocality proof. 
To prove that the quantum distribution remains nonlocal for any possible network structure as in Fig.~\ref{fig:general}(b$_1$), we focus on the network in Fig.~\ref{fig:general}(b$_2$) which leads to the strongest correlations; note that here all parties $A_4$ to $A_N$ come together and act jointly. The main idea of the proof is to consider the conditional distribution $P_Q^{\triangle}(a_1,a_2,a_3)=P_Q(a_1,a_2,a_3|a_4^*,\dots, a_N^*)$ of the outputs of the trusted parties ($A_1$ to $A_3$) conditioned on a specific values $(a_4^*,\dots, a_N^*)$ of the remaining $N-3$ outcomes. Conditioning the source $S_3$ to the outputs $a_4^*,\dots, a_N^*$ is equivalent to having a bipartite source $S_3^*$, so that the distribution $P_Q^{\triangle}(a_1,a_2,a_3)$ corresponds to a triangle network. We can show that there exists an output string $(a_4^*,\dots, a_N^*)$ such that the conditional distribution $P_Q^{\triangle}$ is nonlocal on the triangle. From this, it follows that $P_Q^\bigcirc(a_1,\dots, a_N)$ is nonlocal with respect to the network in Fig. 2(b$_2$). 

Finally, note that conditionally on observing the outputs $a_4^*, \dots, a_N^*$ one can also lower bound the entanglement of $S_1,S_2,$ and  $S_3^*$ and the randomness of $a_1,a_2$ and $a_3$, and show that the measurements of $A_1,A_2$ and $A_3$ are entangled. Moreover, since $a_2$, $S_1$ and $S_2$ are independent of the outcomes of $a_4, \dots, a_N$, the bounds on the randomness of $a_2$ and the entanglement of $S_1$ and $S_2$ are valid independently of the post-selection.  $\square$

It is worth noting that the above proof of nonlocality for  $P^\bigcirc_Q(a_1,\dots,a_N)$ applies for a range of the measurement parameter $u$ that does not decrease with $N$, in contrast to previous results \cite{Renou_2019,Renou2022}.


Let us now move to the second configuration shown in Fig.~\ref{fig:general}(c$_1$). Here we only assume the local network topology for parties $A_1$ and $A_2$. Again, we prove that the quantum distribution $P_Q^\bigcirc(a_1,\dots,a_n)$ on the ring network is nonlocal for any network of the form (c$_1$). This is done by proving nonlocality with respect to network (c$_2$), the strongest network compatible with (c$_1$). 

Before moving on to the results, it is interesting to compare the two configurations (b$_1$) and (c$_1$). One can easily see that any local model compatible with (b$_2$) (and hence with (b$_1$)) is also compatible with (c$_2$), while the opposite does not hold. In this sense, one can argue that configuration (c$_1$) is based on weaker assumptions than (b$_1$). \\


\noindent{\bf Result 3.} For any network structure as in Fig. \ref{fig:general}(c$_1$), i.e. fixing only the local topology of parties $A_1$ and $A_2$, the quantum distribution $P_Q^\bigcirc(a_1,\dots,a_N)$ is nonlocal for $N=4k$ and $N=2k+1$, and for a measurement parameter $u$ close enough to 1.
\\

\noindent\emph{Sketch of the proof.} The full proof is presented in the Appendix~\ref{app: result 2 and 3}. Consider the network in (c$_2$), which is the strongest one compatible with (c$_1$). This is a triangle network, and the fact the coarse-grained distribution is token counting implies that any underlying local model must fulfill some linear constraints discussed in detail in Appendix~\ref{app:sec_triangle}. These constraints can be rewritten as a Linear Program, which turns out to be equivalent to the LP considered in \cite{renou2022network} and shown to be unfeasible for $N=4k$ and $N=2k+1$, and for $u$ close enough to 1. This implies the nonlocality of the quantum distribution with respect to the network in Fig.~\ref{fig:general}(c$_2$). $\square$

\section{Discussion}

Results 2 and 3 show that quantum Bell nonlocality in networks can exhibit strong topological robustness, in the sense that one needs to trust only a very small part of the network structure. Importantly, this trust only concerns the local network topology; no trust is required on the states produced by the sources or on the local measurements, as usual in the study of nonlocality. For example in the configuration of Fig.2(b$_1$), we must trust the local network structure for parties $A_1$ to $A_3$ (as shown in the figure), but we do not require any trust on the sources ($S_1$ and $S_2$) or on the local measurements of $A_1$ to $A_3$, nor on the rest of the network structure for parties $A_4$ to $A_N$.

Moreover, we note that Results 2 and 3 can in fact be applied in parallel to different local neighborhoods of the network. Hence, in the setting of Fig.2(b$_1$), we could have in principle up to $\floor{N/3}$ parties certifying nonlocality and local randomness, independently and at the same time.

For these reasons, we believe that the concept of topologically robust quantum nonlocality may find applications, notably for cryptographic tasks in networks. Here a common scenario is when an adversary takes control over a significant part of the network, potentially jeopardizing the security. We have seen that quantum correlations can be highly resilient against such attacks, allowing notably for the black-box certification of local randomness. 


Our work may also connect to the question of inferring the topology of an a priori unknown network, based on the statistics of local measurements, see e.g. \cite{Aberg2020,Wolfe2021,Kraft2021,Chen2023}. In particular, it would be interesting to see if the methods developed here are relevant for this problem.

From a more fundamental perspective, our work shows that quantum nonlocality can be very robust in the network setting, in the sense that local models can be ruled out even when the constraints on the network structure are relaxed significantly. It is interesting to discuss the connection of these results with previous works focusing on the standard Bell bipartite scenario, where distant parties receive random measurement inputs. Notably, a question that received broad attention is to consider local models with relaxed measurement independence, see e.g.~\cite{brans1988,Hall2010,Puetz2014}. This question can also be naturally phrased in the language of network nonlocality by introducing two additional parties revealing the measurement inputs. As we discuss in Appendix \ref{app: standard Bell}, the resulting distribution only features a very limited topological robustness. Another related direction is the concept of genuine multipartite nonlocality. Here, quantum distributions are compared to local models where parties are allowed to come together in several groups, or share nonlocal resources such as non-signaling correlations, see e.g. ~\cite{Svetlichny,Bancal2013,Chaves2017}.

\section{Conclusion}

In this work, we investigated quantum Bell nonlocality in a setting where the topology of the underlying network is only partially known. We presented several examples of this effect, and a number of methods to address the problem. In particular, we showed that starting from a specific quantum distribution on a ring network with $N$ parties, it is enough to trust the network structure for two or three neighbouring parties in order to guarantee the presence of nonlocality over the entire network. This shows that quantum nonlocality can feature strong topological robustness. Furthermore, we showed that applications of quantum nonlocality, such as black-box certification of randomness and entanglement, can also be achieved in this scenario. The quantum distributions discussed here rely on the concept of token counting and can be realized within quantum optics~\cite{Abiuso2022}. Nevertheless, we expect that topologically robust nonlocality is also possible for other quantum models; in Appendix \ref{Numerical_elegant} we provide numerical evidence. Finally, on the more technical level, we derived more effective self-testing methods for quantum distributions in networks which can be of independent interest.

Our work opens a number of questions. It would be interesting to investigate topologically robust nonlocality in the presence of noise. While we presented some numerical analysis in Appendix \ref{Numerical_noise_robustness}, analytical progress should be possible via the recently introduced technique of approximate rigidity \cite{boreiri2023noise}. Another interesting direction is to discuss topologically robust nonlocality when allowing for partial correlations between the different sources \cite{Supic2020}.

\section{Acknowledgements}

We thank Alex Pozas-Kerstjens for the discussions. We acknowledge financial support from the Swiss National Science Foundation (projects 192244, 214458 and NCCR SwissMAP) and by the
Swiss Secretariat for Education, Research and Innovation (SERI) under contract number UeM019-3. TK was additionally funded by the European Research Council (Consolidator grant ’Cocoquest’ 101043705) and the Austrian Federal Ministry of Education via the Austrian Research Promotion Agency–FFG (flagship project FO999897481, funded by EU program NextGenerationEU).

\bibliographystyle{unsrt}
\bibliography{main.bib}

\onecolumngrid

\setcounter{theorem}{0}



\appendix


\section{Networks, Token counting Correlations and preliminary results}\label{app:preliminary}

\subsection{Notation for networks}

Let us consider a general network with $N$ sources $S_1,\dots,S_N$ and $M$ parties $A_1,\dots A_M$. To define its connectivity, for each party $A_j$ we define the set $\mathds{S}_j\subset[1,N]$ which contains the indices of all the sources connected to it, with the notation $S_i \to A_j$ iff $i\in \mathds{S}_j$. Finally, we are going to discuss probability distributions $P(a_1,\dots,a_K)$ of $K$ outputs labeled $a_k$ for $k\in[1,K]$ observed on different networks. To do so for each party $A_j$ on the network we define the set $\mathds{A}_j \subset [1,K]$ containing the indices of the outputs it produces, with $\mathds{A}_i\cap \mathds{A}_j =\emptyset$ for $i\neq j$ and $\bigcup_i \mathds{A}_i = [1,K].$

A network with $N$ source and $M$ parties is thus a bipartite graph with $N+M$ vertices split as sources and parties. The edges of the graph (connectivity of the network) are represented by the sets $\mathds{S}_1,\dots \mathds{S}_N$ listing all the parties connected to a given source. In addition, to discuss how the correlations of the $K$ outputs are produced on the network we introduced the \emph{disjoint} sets $\mathds{A}_1,\dots, \mathds{A}_M$ specifying which outputs are produced by which party.
To summarize, a network scenario with $M$ parties is fully specified by the following sets 
\be
\mathcal{G}=\{ \mathds{S}_1,\mathds{A}_1,\dots,\mathds{S}_M,\mathds{A}_M\}
\ee
and can be also represented graphically as depicted in..

\subsection{Local models}

Now to define the local correlations produced on a given network, associate a random variable $\lambda_i$ to each source. This random variable is copied and sent to each party $A_j$ connected to $S_i$. The party $A_j$ thus gathers all the variables received from the connected parties connected to it $\bm \lambda^{(j)} = \bigtimes_{i\in \mathds{S}_j} \lambda_i$
here and below $\bigtimes$ denotes the cartesian product with e.g. $\bigtimes_{i=1}^3 x_i=(x_1,x_2,x_3)$, and produces the outcome $\bm a^{(j)} = \bigtimes_{k\in \mathds{A}_j} a_k $
by sampling them from some conditional distribution $p(\bm a^{(j)}|\bm \lambda^{(j)})$. The most general local model thus leads to correlations of the form  
 \be
 P_L(a_1,\dots, a_k) = \mathds{E}\big( p(\bm a^{(1)}|\bm \lambda^{(1)})\dots p(\bm a^{(M)}|\bm \lambda^{(M)}) \big),
 \ee
 with the expected value taken over all the local variables $\lambda_1\dots \lambda_N$ sampled by the sources. 
Note that the functions $\bm a^{(j)}(\bm \lambda^{(j)})$ can be assumed deterministic without loss of generality, i.e. $p(\bm a^{(j)}|\bm \lambda^{(j)}) = 0$ or $1$, since any randomness can be delegated to the sources.

\subsection{Quantum models}
\label{app: quantum models}

\begin{wrapfigure}{r}{0.5\textwidth}
    \begin{center}
    \vspace{0 cm}
     \includegraphics[height=0.25 \textwidth]{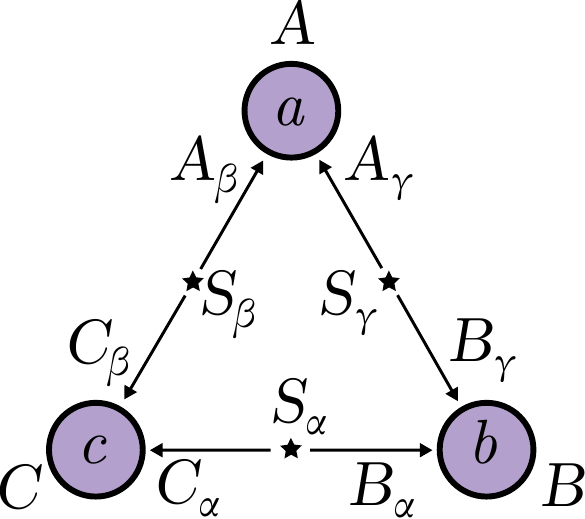}
     \end{center}
       \vspace{-0.4 cm}
    \caption{ Necessary notation to discuss quantum models in the triangle network. The three sources are labeled $S_\alpha,S_\beta, S_\gamma$ and prepare the pairs of systems $B_\alpha C_\alpha$, $C_\beta A_\beta$ and $A_\gamma B_\gamma$ in some bipartite states. The three parties $A,B,C$ produce the outputs $a,b,c$, and we also write $A=A_\gamma A_\beta, B= B_\alpha B_\gamma$ and $C=C_\beta C_\alpha$ to denote the composite quantum systems received by the parties.}
    \label{fig:tri}
\end{wrapfigure}

For a general quantum model, each source $S_i$ produces a multipartite quantum state $\rho_{i}$  and distributes each of the systems composing it to one of the connected parties. Each party $A_j$ gathers the quantum systems received from the connected sources and performs a measurement corresponding to some POVM $\{E^{\bm a^{(j)}}_{A_j}\}$ to produce the outputs $\bm a^{(j)}$. Without loss of generality, one can assume that the POVM is a PVM $(E^{\bm a^{(j)}}_{A_j})^2=E^{\bm a^{(j)}}_{A_j}$, since any POVM can be dilated to a PVM by introducing an auxiliary system, which in turn can be absorbed in one of the sources connected to the party.


We now explicit the notation for quantum models on the triangle, depicted in Fig.~\ref{fig:tri}, as we are going to use it extensively later. Here, we have three sources denoted $S_\alpha, S_\beta, S_\gamma$, and three parties denoted $A,B$ and $C$. Each source prepares the two quantum systems $B_\alpha C_\alpha$, $C_\beta A_\beta$  and $A_\gamma B_\gamma$ respectively in a bipartite quantum state 
\begin{align}
S_\alpha :\,\rho^{(\alpha)}_{B_\alpha C_\alpha}, \qquad
S_\beta : \, \rho^{(\beta)}_{C_\beta A_\beta}, \qquad 
S_\gamma: \, \rho^{(\gamma)}_{A_\gamma B_\gamma },
\end{align}
and distributes them to the parties suggested by the name of the system.

Each party gathers the quantum systems received from the neighbouring sources,  for example, $A_\beta$ and $A_\gamma$ in the case of the party $A$ (we also denote the composed system $A=A_\beta A_\gamma$ slightly abusing the notation), and produces the corresponding output $x$ by performing a POVM measurement $\{E^x_{Y}\}$ where $Y=A,B,C$ labels the party. This gives to the most general quantum correlations in the triangle  $P_Q(a,b,c)= \tr E^a_{A}\otimes E^b_B \otimes E^{c}_C \,\rho^{(\alpha)}_{B_\alpha C_\alpha}\otimes \rho^{(\beta)}_{C_\beta A_\beta} \otimes \rho^{(\gamma)}_{A_\gamma B_\gamma }.$

It is worth noting that each POVM $\{E^x_{Y}\}$ can be dilated to a PVM $\{\Pi^x_{Y'}\}$ by introducing an auxiliary system $D$ prepared in a known state, with $Y'=YD$. Moreover, without loss of generality, the system $D$ can be absorbed in one of the subsystems composing $Y$ and be prepared by the corresponding source, without affecting any of the following claims. Doing this is equivalent to assuming that the measurements are PVM to start with. Without loss of generality, we can thus write any quantum correlations as 
\be
P_Q(a,b,c)= \tr \Pi^a_{A}\otimes \Pi^b_B \otimes \Pi^{c}_C \,\rho^{(\alpha)}_{B_\alpha C_\alpha}\otimes \rho^{(\beta)}_{C_\beta A_\beta} \otimes \rho^{(\gamma)}_{A_\gamma B_\gamma }
\ee
with each party performing a PVM, $(\Pi_Y^x)^2=\Pi_Y^x.$

\subsection{Token Counting Distributions and Rigidity}

Here, we focus on a particular family of strategies on networks, based on classical Token Counting (TC) strategies. In a TC strategy, each source $S_i$ randomly distributes a fixed number of tokens $n_i$ to the neighbouring parties, according to some probability distribution $Q_i(\bm t_i)$. Here, $\bm t_i = \bigtimes_{j\in\mathds{S}_i}t_i^j$  with $\sum_j t_i^j = n_i$ and each $t_i^j$ denotes the number of tokens sent to the party $A_j$ by the source $S_i$. In turn, each party outputs the total number of tokens it received $\tilde {a}_j = \sum_i t_i^j$. The resulting correlations $P(\tilde{a}_1, \dots, \tilde{a}_n)$ are called TC distributions. \\

If the underlying network is \emph{No Double Common-Source} (NDCS), see definition below, the TC distributions are known to be \emph{rigid}, meaning that among all possible classical strategies on the given network, the TC strategy we just described is essentially the unique model leading to $P(\tilde{a}_1, \dots, \tilde{a}_n)$ \cite{Renou2022}. This is formalized in Theorem 1. 

\begin{definition}[No Double Common-Source networks]\label{definition:NDCSnetworks}
A network is called \emph{No Double Common-Source (NDCS)} if each pair of parties does not share more than one common source. More formally, $|\mathds{S}_j \cap \mathds{S}_{j'}|\leq 1$  for $j\neq j'$.
\end{definition}

\begin{theorem}\cite{Renou2022,renou2022network}. \label{th:class rigid} For any classical strategy leading to a TC distribution $P(\tilde{a}_1, \dots, \tilde{a}_n)$ on a NDCS network, for each source $S_i$ there exist local relabeling functions 
\begin{equation}
    T_i^j : \lambda_i \mapsto \mathds{Z}_{\geq 0}
\end{equation}
with the index $j$ running through the parties $A_j$ connected to the source $S_i$, which have the following properties.
\end{theorem}
\begin{align}
(i)\sum_{j| \mathds{S}_i} T_{i}^j(\lambda_i) = n_i,\qquad (ii)\,  \text{Pr} \left(  \bigtimes_{j|S_i\to A_j}T_{i}^{j}(\lambda_i) = \bm t_i \right) = Q_i(\bm t_{i}),\qquad (iii)
\sum_{i|S_i\to A_j} T_{i}^j(\lambda_i) = \tilde{a}_j(\lambda_1,\dots \lambda_N). \nonumber
\end{align}\\

On some networks, the same result, albeit formulated differently, holds for all quantum models \cite{Sekatski2023}. To keep things simple here we only give it for the triangle network as this is the only case we will need later. We will use the notation introduced in Sec.~\ref{app: quantum models} and Fig.~\ref{fig:tri}. Furthermore, we consider the case of $\tilde{a},\tilde{b},\tilde{c} \in\{0,1,2\}$, and consider TC distribution where each source controls exactly $n_\xi=1$ token, and the source $S_\alpha$ sends it to $C$ with probability $p_\alpha$, the source $S_\beta$ sends it to $A$ with probability $p_\beta$, and the source $S_\gamma$ to $B$ with probability $p_\gamma$. The three values $p_\alpha, p_\beta$ and $p_\gamma$ uniquely specify the TC distribution.

\begin{theorem}\cite{Sekatski2023}. \label{th:quant rigid} For any quantum model leading to a TC distribution $P(\tilde a,\tilde b,\tilde c)$ on the triangle network with $\tilde a,\tilde b,\tilde c\in\{0,1,2\}$  and characterized by $p_\alpha, p_\beta$ and $p_\gamma$, the following holds. The states $\rho^{(\alpha)}_{B_\alpha C_\alpha}, \rho^{(\beta)}_{C_\beta A_\beta}, \rho^{(\gamma)}_{A_\gamma B_\gamma}$ prepared by each source admit a purification of the form
\begin{equation}\label{eq: states rigid}
\begin{split}
  &  \ket{\psi_\alpha}_{B_\alpha C_\alpha E_\alpha}= \sqrt{p_\alpha} \ket{01}_{\textbf{B}_\alpha \textbf{C}_\alpha}\!\ket{j^{01}_\alpha}_{J_\alpha}\!\!+ \sqrt{1-p_\alpha} 
\ket{10}_{\textbf{B}_\alpha \textbf{C}_\alpha}\ket{j^{10}_\alpha}_{J_\alpha},\\
 &  \ket{\psi_\beta}_{C_\beta A_\beta E_\beta}= \sqrt{p_\beta} \ket{01}_{\textbf{C}_\beta \textbf{A}_\beta}\!\ket{j^{01}_\beta}_{J_\beta}\!\!+ \sqrt{1-p_\beta} 
\ket{10}_{\textbf{C}_\beta \textbf{A}_\beta}\ket{j^{10}_\alpha}_{J_\beta},\\ &  \ket{\psi_\gamma}_{A_\gamma B_\gamma E_\gamma}= \sqrt{p_\gamma} \ket{01}_{\textbf{A}_\gamma \textbf{B}_\gamma}\!\ket{j^{01}_\gamma}_{J_\gamma}\!\!+ \sqrt{1-p_\gamma} 
\ket{10}_{\textbf{A}_\gamma \textbf{B}_\gamma}\ket{j^{10}_\gamma}_{J_\gamma},\\
\end{split}
\end{equation}
Here the system $J_\alpha = B'_\alpha C'_{\alpha} E_\alpha$ is composed of systems $B'_\alpha$ and $C'_\alpha$ sent to $B$ and $C$ respectively (with $B_\alpha = \bm B_\alpha B'_\alpha$ and $C_\alpha = \bm C_\alpha C'_\alpha$) and a purifying system $E_\alpha$ which is "internal" to the source, and analogous decompositions hold for the remaining systems. The measurements performed by each party $X=A,B,C$ are of the form
\be
\Pi_X^0 = \ketbra{00}_{\textbf{X}_\xi \textbf{X}_{\xi'}}\otimes \mathds{1}_{X_\xi X_\xi'}, \qquad 
\Pi_X^2 =\ketbra{11}_{\textbf{X}_\xi \textbf{X}_{\xi'}}\otimes \mathds{1}_{X_\xi X_\xi'}, \quad
 \Pi_X^1= (\ketbra{01}+\ketbra{10})_{\textbf{X}_\xi \textbf{X}_{\xi'}}\otimes \mathds{1}_{X_\xi X_\xi'} \nonumber,
\ee
where $\xi$ and $\xi'$ denote the sources $S_\xi$ and $S_\xi'$ connected to the party $X$.
\end{theorem}

 This result provides a decomposition of each system $X_\xi$ into a qubit $\bm X_\xi$ and a "junk" $X_\xi'$ on which the measurements act trivially.

\section{Triangle network. Nonlocality, randomness and entanglement}\label{app:sec_triangle}

We consider distributions $P_Q^\triangle( a,b,  c)$ on the triangle network, with the outputs belonging to the sets $\mathds{O}_X= \{0,  1_0, \dots, 1_{|\mathds{O}_X|-3}, 2\}$ for $X=A,B,C$ labeling the parties. The cardinality of each output can be different but we assume that it is larger than three $|\mathds{O}_X|\geq 4$, the minimal set is $\{ 0, 1_0,1_1,2\}$. 
 We assume that after each party performs the deterministic coarse graining $x \mapsto \tilde x \in\{0,1,2\}$ given by  
\be
\tilde x = \begin{cases} 
0&  x =  0\\
1& x =1_i\\
2&  x = 2
\end{cases},
\ee
the induced distribution $P_Q^\triangle(\tilde a, \tilde b, \tilde c)$ is token counting. We now discuss how this property and the rigidity of TC distributions can be combined to prove some properties of the model underlying $P_Q^\triangle(a,b, c)$.

\subsection{Nonlocality}

By "classical" rigidity (Theorem~\ref{th:class rigid}) this observation puts severe constraints on any local model leading to $P(\bar a,\bar b, \bar c)$. Following \cite{RenouTriangle1} one can prove that the distribution is nonlocal by showing that these constraints are infeasible on the underlying network. We formulate this result in the form inspired by \cite{boreiri2023noise}.

To state the theorem let us introduce what we call semi-local models. For such a model each source $S_\xi$ prepares a quantum-classical state, which for simplicity we can write as  
\be\label{eq: semilocal state}
\varrho_\xi =  \sum_{\xi} p(\xi) \ketbra{\xi}_C^{\otimes m} \otimes \rho^{(\xi)}_{Q_1\dots Q_m},
\ee
Here  $Q_1,\dots,Q_m$ denote the quantum systems with $m$ being the number of parties connected to the source. $C$ denotes the classical subsystem with all the states $\ket{\xi}_C$ being orthogonal, a copy of the random variable $\lambda$ is sent to each of the parties connected to the source. Note that we write the sum over $\ketbra{\xi}$ for convenience, more generally one may think that the source first samples a random variable $\xi$, and then a quantum state $\rho^{(\xi)}_{Q_1\dots Q_m}$ conditional on its value. Furthermore, this state is assumed quantum for simplicity, but it may be in principle described with another generalized probability theory without affecting the following arguments.

Upon receiving all the classical and quantum systems from the connected sources each party $X$ performs a measurement, which can be described by a set of PVMs 
\be
\mathcal{M}_X=\Pi_X^{\bar {x} |\xi,\xi'},
\ee
where $\xi, \xi'$ are the local variables sampled by the sources connected to $X$. \\

\noindent \textbf{Definition.} We say that the distribution $P^\triangle(a,b,c)$ is semi-local with respect to the (coarse-grained) outputs $\tilde a,\tilde b,\tilde c$, if it can be reproduced by a semi-local model $\varrho_\alpha, \varrho_\beta, \varrho_\gamma, \mathcal{M}_A,  \mathcal{M}_B, \mathcal{M}_C$ such that the values of  $\tilde a,\tilde b,\tilde c$ are determined by the local variables $\alpha, \beta,\gamma$, i.e. $\Pi^{x|\xi,\xi'}_X= x(\xi,\xi') \, \id_{\bm Q}$.

\begin{theorem} \label{th: semi-local} For a distribution $P^\triangle(a,b,c)$ which is a semi-local with respect to the coarse-grained outputs $\tilde a,\tilde b,\tilde c$, and such that the distribution $P^\triangle(\tilde a,\tilde b,\tilde c)$ is token counting,  the following linear program is feasible
\begin{align}\label{eq: feasibility locality}
\min_{\substack{ q(i,j,k|\cw) \\ q(i,j,k|\acw)}}  \qquad &1 \\
\textrm{such that}\qquad &q(i,j,k|\circlearrowright),\,q(i,j,k|\acw)\geq 0, \quad \sum_{i,j,k} q(i,j,k|\cw)=\sum_{i,j,k} q(i,j,k|\acw)=1\label{th3: c1}
\\
 & q(i|\cw)= \frac{P^\triangle(1_i,2,0)}{P^\triangle(1,2,0)}  \text{ and cyclic permutations} \label{th3: c3}
 \\
& q(i|\acw)= \frac{P^\triangle(1_i,0,2)}{P^\triangle(1,0,2)} \text{ and cyclic permutations} \label{th3: c4} 
\\
&q(i,j,k|\cw)\, p_\alpha p_\beta p_\gamma+ q(i,j,k|\acw)\, (1-p_\alpha)(1-p_\beta)(1-p_\gamma) = P^\triangle(1_i, 1_j, 1_k) \label{th3: c2}\\
&p_\alpha = \frac{P^\triangle(0,1,2)}{P^\triangle(2,0,1)+P^\triangle(0,1,2)},\, p_\beta = \frac{P^\triangle(2,0,1)}{P^\triangle(1,2,0)+P^\triangle(2,0,1)},\, p_\gamma = \frac{P^\triangle(1,2,0)}{P^\triangle(0,1,2)+P^\triangle(1,2,0)}.\label{eq: papbpc}
\end{align}
\end{theorem}

\begin{proof}
Since $P^\triangle(a, b, c)$ is semi-local and the coarse-grained distribution $P^\triangle(\tilde a,\tilde b,\tilde c)$ is TC, applying theorem~\ref{th:class rigid} guarantees the existence of the relabeling functions $T_{\alpha}^B,T_{\alpha}^C,T_{\beta}^C,T_{\beta}^A,T_{\gamma}^A, T_{\gamma}^B \in \{0,1\}$, applied on the respective local variables $\alpha,\beta$ and $\gamma$, and such that 
\begin{align}\label{eq: T+T}
&T_\alpha^B+T_\alpha^C =1 \qquad \qquad \, \, T_\beta^C+T_\beta^A =1 \qquad \qquad \, \, T_\gamma^A+T_\gamma^B =1 \\
\label{eq: rig rig rig}
 &T_\gamma^A +T_\beta^A=\tilde a(\beta,\gamma) \qquad 
 T_\alpha^B+T_\gamma^B = \tilde b(\alpha,\gamma)  \qquad 
T_\alpha^C+T_\beta^C = \tilde c(\alpha,\beta).
\end{align}
To shorten the notation it is convenient to introduce the binary variables $t_\alpha = T_{\alpha}^C$, $t_\beta = T_{\beta}^A$ and $t_\gamma =T_{\gamma}^B$ with $T_{\alpha}^B= 1-t_\alpha,\dots$ by Eq.~\eqref{eq: T+T}. The coarse-grained outputs $\tilde a,\tilde b,\tilde c$ are fully determined by the values $\bm t= (t_\alpha,t_\beta, t_\gamma)$ sampled at the sources.
Except the cases $\bm t =(0,0,0)$ and $\bm t =(1,1,1)$, which both result in $(\tilde a,\tilde b,\tilde c)=(1,1,1)$ all combinations or output correspond to a unique assignment
\begin{align}
\bm t &= (0,1,1) \implies (\tilde a,\tilde b,\tilde c) = (1,2,0) \quad \text{and cyclic permutations}\\
\bm t &= (1,0,0) \implies (\tilde a,\tilde b,\tilde c) = (1,0,2) \quad \text{and cyclic permutations}.
\end{align}
Rigidity also fixes the distribution of the token variables $p_\xi = \text{Pr}(t_\xi=1)$ in Eq.~\eqref{eq: papbpc}, in particular from $(1-p_\alpha)p_\beta p_\gamma= P^\triangle(1,2,0)$ and $(1-p_\alpha)p_\beta (1-p_\gamma) = P^\triangle(2,1,0)$ one easily verifies  that
$p_\gamma = \frac{P(1,2,0)}{P(1,2,0)+P(2,1,0)}.$\\

To shorten notation let us introduce a label for $\bm t =(1,1,1)=\cw$ and $(0,0,0)=\acw$, corresponding to all the tokens sent clockwise or anticlockwise in Fig.~\ref{fig:tri}. The probabilities of these events are given by
\begin{align}
\text{Pr}(\cw)=\text{Pr}(\bm t =\cw)= p_\alpha p_\beta p_\gamma   
\qquad
\text{Pr}(\acw)=\text{Pr}(\bm t = \acw)=  (1-p_\alpha)(1-p_\beta)(1-p_\gamma).
\end{align}

Note that the variables of the full outputs $a,b,c\in\{0,1_\ell,2\}$ are not in general determined by $\bm t$.  More precisely,  this is not the case of the index respective indexes $i,j,k$ in $ a= 1_i$, $ c= 1_j$, $ b= 1_k$, when $\tilde a,\tilde b,$ or $\tilde c$ equals 1, the. Nevertheless, there exists a (hidden) probability distribution
$P(a,b,c,\bm t)$
describing the joint distribution of the full outcomes and the token variables. It allows us to define the variables
\be
\begin{split}
q(i,j,k|\cw)&= \text{Pr}(a=1_i,b=1_j,c=1_k| \bm t=\cw)=
\frac{\text{Pr}(a=1_i,b=1_j,c=1_k, \bm t=\cw)}{\text{Pr}(\bm t=\cw)} \\
q(i,j,k|\acw)&= \text{Pr}(a=1_i,b=1_j,c=1_k|\bm t=\acw)= \frac{\text{Pr}(a=1_i,b=1_j,c=1_k, \bm t=\acw)}{\text{Pr}(\bm t=\acw)}
\end{split}
\ee
We can readily see that they by definition satisfy the constraints \eqref{th3: c1}, \eqref{th3: c2} and \eqref{eq: papbpc}. Now let us show show that Eqs.~(\ref{th3: c3},\ref{th3: c4}) must also hold.
To do so we focus on 
\be
q(i|\cw) = \text{Pr}(a=1_i|\cw) = \text{Pr}(a=1_i|\bm t =(1,1,1)),
\ee
the network structure guarantees that the output $a$ is independent of the value of local variable $\alpha$ (and $t_\alpha$). Hence we find that
\be
\text{Pr}(a=1_i|\bm t =(1,1,1))= 
\text{Pr}(a=1_i|\bm t =(0,1,1)) = \text{Pr}(1_i,2,0|\bm t =(0,1,1)).
\ee
In addition, the token assignment $\bm t =(0,1,1)$ corresponds to a unique coarse-grained outputs $(\tilde a,\tilde b,\tilde c)=(1,2,0)$, therefore $\text{Pr}(\bm t =(0,1,1)) = P^\triangle(1,2,0)$ and 
\be 
\text{Pr}( 1_i,2,0|\bm t =(0,1,1)) = \frac{\text{Pr}(1_i,2,0,\bm t =(0,1,1))}{\text{Pr}(\bm t =(0,1,1))}= 
\frac{P^\triangle(1_i,2,0)}{P^\triangle(1,2,0)}.
\ee
Combining everything together gives 
\be
q(i|\cw)= \frac{P^\triangle(1_i,2,0)}{P^\triangle(1,2,0)}
\ee
The other constraint (\ref{th3: c3},\ref{th3: c4}) can be proven in the same way.
This shows that under the assumptions of the theorem there must exist the variables $q(i,j,k|\cw)$ and $q(i,j,k|\acw)$ satisfying all the constraints of the LP. 
\end{proof}

\begin{theorem}\label{the: LP nonloc tri}
Consider a distribution $P^\triangle( a, b,  c)$, such that the coarse-grained distribution $P^\triangle(\tilde a,\tilde b,\tilde c)$ is token counting. $P^\triangle(a,b, c)$ is nonlocal if the LP in equation \eqref{eq: feasibility locality} is infeasible.
\end{theorem}

\begin{proof} Let the LP be infeasible and assume that $P^\triangle(a,b, c)$ is local. Then it is also semi-local for the outputs $\tilde a,\tilde b,\tilde c$. Hence, Theorem 3 implies that the LP~ \eqref{eq: feasibility locality} is feasible, proving the corollary by contradiction.
\end{proof}

\subsection{Global coherence}

Here our goal is not only to prove the nonlocality of $P^\triangle_Q(a,b,c)$, but also to obtain some quantitative statements on the underlying \textit{quantum} model.  We thus consider a general quantum model leading to the $P_Q(a,b, c)$. As discussed already it involves some projective quantum measurement $\{\Pi_X^{ x}\}$ for each party. The coarse-graining of the outputs is now described as a summation of the corresponding POVM elements, giving rise to the coarse-grained PVMs with three elements $\{\Pi_X^0,\Pi_X^1,\Pi_X^2\}$.

Now, following \cite{Sekatski2023} we use the "quantum" rigidity (Theorem \ref{th:quant rigid}) to derive quantitative bounds on all quantum models that could lead to the observed distributions. In particular, it implies that the coarse-grained measurements of each party are of the form
\begin{align} \label{eq: Pi0}
\Pi^{0}_X &= {\Pi}^{\bar 0}_X= \ketbra{00}_{\textbf{X}_\xi \textbf{X}_{\xi'}}\otimes \mathds{1}_{X_\xi X_\xi'}\\ \label{eq: Pi1}
 \Pi_X^1 &= \sum_{i} {\Pi}^{ 1_i}_X=(\ketbra{01}+\ketbra{10})_{\textbf{X}_\xi \textbf{X}_{\xi'}}\otimes \mathds{1}_{X_\xi X_\xi'}
\\
\label{eq: Pi2}
\Pi_X^2 &= \Pi^{\bar 2}_X = \ketbra{11}_{\textbf{X}_\xi\textbf{X}_{\xi'}}\otimes \mathds{1}_{X_\xi X_\xi'}.
\end{align}
In turn, it also guarantees that the global quantum state  prepared by the three source can be purified to 
\be\label{eq: psi xi}
\ket{\Psi} = \ket{\psi_\alpha}_{B_\alpha C_\alpha E_\alpha} \ket{\psi_\beta}_{C_\beta A_\beta E_\beta} \ket{\psi_\gamma}_{A_\gamma B_\gamma E_\gamma}, \quad \text{with} \quad 
   \ket{\psi_\xi}_{X_\xi Y_\xi E_\xi}= \sqrt{p_\xi} \ket{01}_{\textbf{X}_\xi \textbf{Y}_\xi}\!\ket{j^{01}_\xi}_{J_\xi}\!\!+ \sqrt{1-p_\xi} 
\ket{10}_{\textbf{X}_\xi \textbf{Y}_\xi}\ket{j^{10}_\xi}_{J_\xi}
\ee
the individual states given in Eq.~\eqref{eq: states rigid} in full detail, with the distribution given by  $P_Q(\bar a, \bar b, \bar c) = \bra{\Psi} \Pi_A^{\bar a} \Pi_B^{\bar b}\Pi_C^{\bar c}\ket{\Psi}$.
To shorten the notation we define the global states
\begin{align}\label{app: p1}
\ket{\cw}  &=\ket{01,01,01}_{\textbf{B}_\alpha \textbf{C}_\alpha \textbf{C}_\beta \textbf{A}_\beta \textbf{A}_\gamma \textbf{B}_\gamma}\ket{j_\alpha^{01},j_\beta^{01},j_\gamma^{01}}_{J_\alpha J_\beta J_\gamma} \\
\ket{\acw}   &=\ket{10,10,10}_{\textbf{B}_\alpha \textbf{C}_\alpha \textbf{C}_\beta \textbf{A}_\beta \textbf{A}_\gamma \textbf{B}_\gamma}\ket{j_\alpha^{10},j_\beta^{10},j_\gamma^{10}}_{J_\alpha J_\beta J_\gamma}.
\end{align}

One easily sees from the Eq.~\eqref{eq: Pi1} that the global projector $ \Pi_A^{1}\Pi_B^{1}\Pi_C^{1}$ corresponding to the outcome $(\tilde a,\tilde b,\tilde c)=(1,1,1)$ is only nonzero on the part of the state $\ket{\Psi}$ in the subspace spanned by $\ket{\cw}$ and $\ket{\acw}$, that is  
\begin{align} \label{eq: pr 111}
P_Q^\triangle(1,1,1) &= \bra{\Psi} \Pi_A^{1}\Pi_B^{1}\Pi_C^{1} \ket{\Psi} =  \bra{\Psi}\left( \ketbra{\cw}+ \ketbra{\acw}\right) \Pi_A^{1}\Pi_B^{1}\Pi_C^{1} \left( \ketbra{\cw}+ \ketbra{\acw}\right)\ket{\Psi} \\
& = \left(\sqrt{\text{Pr}(\cw)} \bra{\cw} + \sqrt{\text{Pr}(\acw)} \bra{\acw} \right) \Pi_A^{1}\Pi_B^{1}\Pi_C^{1} \left(\sqrt{\text{Pr}(\cw)} \ket{\cw} + \sqrt{\text{Pr}(\acw)} \ket{\acw}\right)
\\
& = \text{Pr}(\cw) \bra{\cw}  \Pi_A^{1}\Pi_B^{1}\Pi_C^{1} \ket{\cw} + \text{Pr}(\acw) \bra{\acw}  \Pi_A^{1}\Pi_B^{1}\Pi_C^{1} \ket{\acw}, \qquad  \text{with} \label{eq: P111}\\ 
\text{Pr}(\cw) &=  p_\alpha p_\beta p_\gamma,\\
\text{Pr}(\acw)&= (1-p_\alpha)(1-p_\beta)(1-p_\gamma),
\end{align}
where we used $\Pi_A^{1}\Pi_B^{1}\Pi_C^{1} \ket{\cw} = \ket{\cw}$ which is orthogonal to $\ket{\acw}$. Now, for all outcomes $( a,  b,  c)=( 1_i, 1_j, 1_k)$ we have $\Pi_A^{ 1_i}\Pi_B^{ 1_j}\Pi_C^{ 1_k}\leq \Pi_A^{1}\Pi_B^{1}\Pi_C^{1}$ and thus 
\begin{align}
P_Q^\triangle( 1_i, 1_j, 1_k) &=  \text{Pr}(\acw) \, \bra{\acw} \Pi_A^{1_i}\Pi_B^{ 1_i}\Pi_C^{ 1_k} \ket{\acw} + \text{Pr}(\cw) \bra{\cw} \Pi_A^{1_i}\Pi_B^{1_i}\Pi_C^{ 1_k} \ket{\cw} + 2 \sqrt{\text{Pr}(\acw)\text{Pr}(\cw)} \,  r_{ijk} \\
\text{where} \quad &r_{ijk} := 
\, \text{Re} \bra{\acw} \Pi_A^{ 1_i}\Pi_B^{ 1_i}\Pi_C^{ 1_k}\ket{\cw} .
\label{app: pf}
\end{align}
By the Cauchy-Schwarz inequality the coherence term $r_{ijk}$ must also satisfy
\be\label{eq: new bound r}
|r_{ijk}|=|\text{Re} \bra{\acw} \Pi_A^{ 1_i}\Pi_B^{ 1_i}\Pi_C^{ 1_k}\ket{\cw}| \leq | \bra{\acw} \Pi_A^{ 1_i}\Pi_B^{ 1_i}\Pi_C^{ 1_k}\ket{\cw}| \leq \sqrt{\bra{\acw} \Pi_A^{1_i}\Pi_B^{ 1_i}\Pi_C^{ 1_k} \ket{\acw} \bra{\cw} \Pi_A^{1_i}\Pi_B^{ 1_i}\Pi_C^{ 1_k} \ket{\cw}}.
\ee

In addition, from $P_Q^\triangle(1,1_j,1_k)=\sum_{i} P^\triangle_Q( 1_i,1_j, 1_k)$ we obtain
\begin{align}
P_Q(1,1_j,1_k) &= \left(\sqrt{\text{Pr}(\cw)} \bra{\cw} + \sqrt{\text{Pr}(\acw)} \bra{\acw} \right) \Pi_A^{1}\Pi_B^{1_j}\Pi_C^{1_k} \left(\sqrt{\text{Pr}(\cw)} \ket{\cw} + \sqrt{\text{Pr}(\acw)} \ket{\acw}\right)\\
&= \text{Pr}(\cw) \bra{\cw } \Pi_A^{1}\Pi_B^{1_j}\Pi_C^{1_k}  \ket{\cw} + \text{Pr}(\acw) \bra{\acw}  \Pi_A^{1}\Pi_B^{1_j}\Pi_C^{1_k} \ket{\acw}
\end{align}
by virtue of $ \bra{\acw }\Pi_A^{1}\ket{\cw} = \bra{01}_{A_\beta A_\gamma}\Pi_A^{1}\ket{10}_{A_\beta A_\gamma} =0$, and
\be
\sum_{i} P_Q^\triangle( 1_i, 1_j, 1_k) = \text{Pr}(\cw) \bra{\cw } \Pi_A^{1}\Pi_B^{1_j}\Pi_C^{1_k}  \ket{\cw} + \text{Pr}(\acw) \bra{\acw}  \Pi_A^{1}\Pi_B^{1_j}\Pi_C^{1_k} \ket{\acw} +  2 \sqrt{\text{Pr}(\cw)\text{Pr}(\acw)} \sum_i r_{ijk},
\ee
which imply the following constraints on the coherences
\be
\sum_i r_{ijk} = \sum_j r_{ijk} =\sum_k r_{ijk} =0.
\ee

Let us now also express the remaining probabilities using the rigidity of the quantum model. Let us focus on the term $ P_Q^\triangle(1_i,0,2)=\| \Pi^{1_i}_A \Pi^{0}_B\Pi^{2}_C \ket{\Psi}\|^2$ given by 
\be \label{eq: pr 102}
\begin{split}
    P_Q^\triangle(1_i,0,2) &=\| \Pi_A^{1_i} (\ketbra{00}_{\textbf{B}_\gamma \textbf{B}_\alpha}\otimes \mathds{1}_{J_B})(\ketbra{11}_{\textbf{C}_\alpha \textbf{C}_\beta}\otimes \mathds{1}_{J_C})\ket{\Psi}\|^2 \\
    &= p_\alpha (1-p_{\beta}) (1-p_{\gamma}) \ \| \Pi^{1_i}_A \ket{01,10,10}_{\textbf{B}_\alpha \textbf{C}_\alpha \textbf{C}_\beta \textbf{A}_\beta  \textbf{A}_\gamma  \textbf{B}_\gamma}\ket{j_\alpha^{01},j_\beta^{10},j_\gamma^{10}}_{J_\alpha J_\beta J_\gamma}\|^2 \\
    & = p_\alpha (1-p_{\beta}) (1-p_{\gamma}) \| \Pi^{1_i}_A \ket{01}_{\textbf{A}_\beta  \textbf{A}_\gamma}\ket{j_\beta^{10},j_\gamma^{10}}_{J_\beta J_\gamma}\|^2.
\end{split}
\ee
As before, this term is only consistent with a unique assignment of the token's directions and hence does not involve any coherence $r_{ijk}$.  Similar expressions for the outputs $(1_i,0,2)$ and permutations thereof can be computed in the same manner.

Finally, one notices that the coherence terms $r_{ijk }$ are genuinely quantum, as they come from the interference of all the tokens going clockwise or anticlockwise on the ring. Moreover, as we now discuss, showing that at least one of $r_{ijk}$ is nonzero implies that the distribution is nonlocal. In addition, a lower bound on its value can be used to quantify the entanglement of the sources and randomness of the outputs as discussed in the following. First, let us show how the value $r_{ijk}$ can be bounded systematically, the following lemma gives the key idea.\\

\begin{theorem} \label{th: 4}Consider any quantum model on the triangle with the measurements $\Pi_X^{x}$ satisfying Eqs.~(\ref{eq: Pi0}-\ref{eq: Pi2}) and the states $\ket{\psi_\xi}$ of the form of Eq.~\eqref{eq: psi xi}, leading to a distribution $P_Q^\triangle( a,  b, c)$.  The distribution
\be
Q^\triangle( a, b,  c) := 
\begin{cases}
P_Q^\triangle({1}_i,  {1}_j,  {1}_k) -2 \sqrt{\text{Pr}(\cw)\text{Pr}(\acw)} r_{ijk} & \tilde a=\tilde b=\tilde c=1 \\
P_Q^\triangle(a, b,  c) & \text{otherwise}
\end{cases},
\ee
with $\text{Pr}(\cw)\text{Pr}(\acw)$ and $r_{ijk}$ defined in Eqs.~(\ref{app: p1}-\ref{app: pf}), is semi-local with respect to the outputs $\tilde a,\tilde b,\tilde c$. Therefore, $Q^\triangle( a, b,  c)$ fulfills the assumptions of Theorem~\ref{th: semi-local} and the induced LP~\eqref{eq: feasibility locality} is feasible.
\end{theorem}

\begin{proof} To show this consider a quantum model where the measurements are kept the same, while the states are replaced by the states 
\be\label{eq: state tilde rho}
 \tilde \rho^{(\xi)}_{X_\xi Y\xi E_\xi} = p_\xi \ketbra{01}_{\textbf{X}_\xi \textbf{Y}_\xi} \otimes\ketbra{j_\xi^{01}}_{ J_\xi} + (1-p_\xi) \ketbra{10}_{\textbf{X}_\xi \textbf{Y}_\xi}\otimes \ketbra{j_\xi^{10}}_{ J_\xi},
\ee
where the token directions have been "decohered". First, let us show that this model leads to the distribution $Q^\triangle (\bar a,\bar b, \bar c)$. As we have argued the probabilities of the outputs with $\tilde a,\tilde b$ or $\tilde c \neq 1$ in Eq.~\eqref{eq: pr 102} are independent of $r_{ijk}$. Hence the probabilities of these outputs are left unchanged when decohering the token directions
\be
Q^\triangle (1_i,0, 2) = 
P_Q^\triangle (1_i,0, 2)  \qquad \text{and permutations}.
\ee
In turn, by setting $r_{ijk}=0$ in Eq.~\eqref{eq: pr 111} we obtain
\be
P_Q^\triangle(1_i, 1_j, 1_k) -2 \sqrt{\text{Pr}(\cw)\text{Pr}(\acw)} r_{ijk} =  \text{Pr}(\acw) \, \bra{\acw} \Pi_A^{ 1_i}\Pi_B^{ 1_i}\Pi_C^{ 1_k} \ket{\acw} + \text{Pr}(\cw) \bra{\cw} \Pi_A^{ 1_i}\Pi_B^{ 1_i}\Pi_C^{1_k} \ket{\cw} = Q^\triangle(1_i,1_j,1_k),
\ee
showing that $Q^\triangle(a,b,c)$ is indeed obtained from the original quantum model by decohering the token directions.

Next, let us show that this model is semi-local for the outputs $\tilde a,\tilde b,\tilde c$. In fact, there is nothing to show here, as in the states $\tilde \rho^{(\xi)}_{X_\xi Y\xi E_\xi}$ in Eq.\eqref{eq: state tilde rho} the system $\textbf{X}_\xi \textbf{Y}_\xi$ are explicitly classical, and the outputs $\tilde a,\tilde b,\tilde c$ only depend on these systems as can be seen from Eqs.~(\ref{eq: Pi0}-\ref{eq: Pi2}).
\end{proof}

This theorem can be readily used to lower bound the coherences $r_{ijk}$. For instance, let us consider the quantity
\be
 \sum_{ijk} (-1)^{s_{ijk}} r_{ijk} 
\ee
for some fixed assignment of signs $s_{ijk}\in\{0,1\}$. It can be bounded with the following linear program.

\begin{theorem} \label{th: Quantum LP}
    Let $P^\triangle_Q(a,b,c)$ be a quantum distribution on the triangle, such that the coarse-grained distribution $P^\triangle(\tilde a,\tilde b,\tilde c)$ is token counting. The underlying model is thus of the form given in Eqs.~(\ref{eq: Pi0}-\ref{app: pf}), and such that the coherences $r_{ijk}$ satisfy $\sum_{ijk} (-1)^{s_{ijk}} r_{ijk} \geq R$ for some fixed signs $s_{ijk}\in\{0,1\}$ where
\end{theorem}
\begin{align}
\label{eq: optimization r}
R = \min_{r_{000},\dots} \quad &\sum_{ijk} (-1)^{s_{ijk}} r_{ijk}  \\
\textrm{such that} \quad 
& \sum_i r_{ijk}= \sum_j r_{ijk} = \sum_k r_{ijk} = 0
\\ \label{eq: tilde P def}
& Q^\triangle(a,b,c) =
\begin{cases} P_Q^\triangle(a,b,c) -2 \sqrt{p_\alpha p_\beta p_\gamma (1-p_\alpha)(1-p_\beta)(1-p_\gamma)}  r_{ijk} & (a, b, c) =( 1_i,1_j,1_k)
\\
P_Q^\triangle(a,  b,  c)  &  (a, b, c) \neq ( 1_i,1_j,1_k)
\end{cases}
\\ 
&p_\alpha = \frac{P^\triangle(0,1,2)}{P^\triangle(0,1,2)+P^\triangle(0,2,1)},\, p_\beta = \frac{P^\triangle(2,0,1)}{P^\triangle(2,0,1)+P^\triangle(1,0,2)},\, p_\gamma = \frac{P^\triangle(1,2,0)}{P^\triangle(1,2,0)+P^\triangle(2,1,0)}\\
& Q^\triangle (\bar a, \bar b, \bar c) \text{ is feasible in linear program of Eq.~\eqref{eq: feasibility locality}} 
\end{align}
\begin{proof}
As discussed in the beginning of the section, the fact that the underlying model is of the form~(\ref{eq: Pi0}-\ref{app: pf}) is guaranteed by the quantum token counting rigidity. It also implies $\sum_i r_{ijk}= \sum_j r_{ijk} = \sum_k r_{ijk} = 0$, albeit we do not know the values of $r_{ijk}$. In turn, this allows un to formally define the probability distributions $ P(\bar a, \bar b, \bar c)$ via the Eq.~\eqref{eq: tilde P def}. Finally, Theorem ~\ref{th: 4} 
guarantees that this must be a valid probability distribution which fulfills the LP in Eq.~\eqref{eq: feasibility locality}. It remains to search through all possible values of the variable $r_{000},r_{001}\dots$ to find the one minimizing the goal function $\sum_{ijk} (-1)^{s_{ijk}} r_{ijk}$. This is precisely what the optimization does.
\end{proof}

\begin{remark}
    A particularly simple situation is when all the indices $i,j,k$ are binary, guaranteeing $r_{0jk}=-r_{1jk}$, $r_{i0k}=-r_{i1k}$, $r_{ij0}=-r_{ij1}$ and $|r_{ijk}|=|r_{000}|$. In this case up to a sign there is a unique variable, say $r_{000}$, to minimize (one is free to maximize $r_{000}$ or $-r_{000}$ and choose the best result).  

\end{remark}


\subsection{Randomness and entanglement}
\label{app: randomness and entanglement}
In the previous section, we have argued how the global coherence 
\begin{align}\label{eq: R bound}
R &\leq \sum_{ijk} (-1)^{s_{ijk}} r_{ijk} = \text{Re} \bra{\acw} \sum_{ijk} (-1)^{s_{ijk}}  \Pi_A^{ 1_i}\Pi_B^{ 1_i}\Pi_C^{ 1_k}\ket{\cw} = \text{Re} \bra{\acw} W_{ABC} \ket{\cw}
\end{align}
where we introduced the Hermitian operator $W_{ABC}=\sum_{ijk} (-1)^{s_{ijk}}  \Pi_A^{\bar 1_i}\Pi_B^{\bar 1_i}\Pi_C^{\bar 1_k}$ whose eigenvalues are between -1 and 1,
between all tokens sent clockwise and anticlockwise can be bounded. We now discuss what this bound implies on the characteristics of the underlying quantum model. We are interested in the entanglement of states and measurements, and the randomness of outcomes, closely following the derivations in \cite{Sekatski2023}.  To start let us rewrite the bound ~\eqref{eq: R bound} in a form easier to use and interpret. We have
\begin{align} \label{eq: R bound 2}
R&\leq  \text{Re} \tr (\tr_{E} \ketbra{\cw}{ \acw} ) W_{ABC} \leq \|(\tr_{E} \ketbra{\cw}{ \acw} ) W_{ABC} \|_1 
\leq \|\tr_{E} \ketbra{\cw}{ \acw} \|_1 \| W_{ABC} \|_\infty = \|\tr_{E} \ketbra{\cw}{ \acw} \|_1\\
&= \left \|(\tr_{E_\alpha} \ketbra{01,j^{01}_\alpha}{10,j^{10}_\alpha}_{B_\alpha C_\alpha E_\alpha})\otimes
(\tr_{E_\beta} \ketbra{01,j^{01}_\beta}{10,j^{10}_\beta}_{C_\beta A_\beta E_\beta})\otimes
(\tr_{E_\gamma} \ketbra{01,j^{01}_\gamma}{10,j^{10}_\gamma}_{A_\gamma B_\gamma E_\gamma})  \right\|_1 \nonumber
\end{align}
for the 1-Schatten norm $\|A\|_1 = \tr \sqrt{A^\dag A}$.

First, for each of the states $\rho^\xi_{X_\xi Y_\xi } = \tr_{E_\xi} \ketbra{\psi_\xi}_{X_\xi Y_\xi E_\xi}$
received by the parties, Eq.~\eqref{eq: R bound 2} can be used to bound the coherence 
\begin{align}
R\leq \|\tr_{E_\xi} \ketbra{01,j^{01}_\xi}{10,j^{10}_\xi}_{X_\xi Y_\xi E_\xi})  \|_1\qquad \forall \,\xi
\label{eq: norm bound}
\end{align}
between the tokens going left or right. In turn, this can be used to bound the entanglement between the systems $X_\xi$ and $Y_\xi$.

Second, we can also use Eq.~\eqref{eq: R bound 2} to imply
\begin{align} 
R&\leq   \left \|(\tr_{E_\alpha} \ketbra{j^{01}_\alpha}{j^{10}_\alpha}_{B_\alpha C_\alpha E_\alpha})\otimes
(\tr_{E_\beta} \ketbra{j^{01}_\beta}{j^{10}_\beta}_{C_\beta A_\beta E_\beta})\otimes
(\tr_{E_\gamma} \ketbra{j^{01}_\gamma}{^{10}_\gamma}_{A_\gamma B_\gamma E_\gamma})  \right\|_1\nonumber \\
& = \max_{U,V,W} \big(\tr U_{B_\alpha C_\alpha} (\tr_{E_\alpha} \ketbra{j^{01}_\alpha}{j^{10}_\alpha}_{B_\alpha C_\alpha E_\alpha}) \big) 
\big(\tr V_{C_\beta A_\beta}(\tr_{E_\beta} \ketbra{j^{01}_\beta}{j^{10}_\beta}_{C_\beta A_\beta E_\beta}) \big)
\big(\tr W_{A_\gamma B_\gamma} (\tr_{E_\gamma} \ketbra{j^{01}_\gamma}{j^{10}_\gamma}_{A_\gamma B_\gamma E_\gamma})\big) \nonumber \\
& = \max_{U,V,W} \bra{j_\alpha^{10}} U_{B_\alpha C_\alpha}\otimes \id_{E_\alpha} \ket{j_\alpha^{01}}
\bra{j_\beta^{10}} V_{C_\beta A_\beta}\otimes \id_{E_\beta} \ket{j_\beta^{01}}
\bra{j_\gamma^{10}} W_{A_\gamma B_\gamma}\otimes \id_{E_\gamma} \ket{j_\gamma^{01}}\nonumber\\
& = F(\rho_{E_\alpha}^{01},\rho_{E_\alpha}^{10}) F(\rho_{E_\beta}^{01},\rho_{E_\beta}^{10}) F(\rho_{E_\gamma}^{01},\rho_{E_\gamma}^{10}) \label{eq: bound fid}
\end{align}
with the marginal state of the purifying systems $\rho_{E_\xi}^{01} = \tr_{X'_\xi Y'_\xi} \ketbra{j_\xi^{01}}_{X'_\xi Y'_\xi E_\xi}$ (and same for $\rho_{E_\xi}^{10}$). Here we used the Uhlmann's theorem stating that the fidelity between two states $F(\rho,\sigma)=  \|\sqrt{\rho} \sqrt{\sigma}\|_1$ is given by the maximal overlap between their purifications. This inequality can be used to bound the randomness of outcomes.

\subsubsection{All measurements are entangled}

 To show that all the measurements must be entangled, the proof of \cite{Sekatski2023} can be applied. The idea is that the operators $\Pi_{X}^{{1}_l}$ define a POVM on the subspace with one token, projected by $\Pi_X^{1}=(\ketbra{01}+\ketbra{10})_{\bm X_\xi \bm X_{\xi'}}\otimes \id$. Thus, the separable eigenstates of $\Pi_{X}^{\bar{1}_l}$ must thus project on $\ketbra{01}_{\bm X_\xi \bm X_{\xi'}}$ or $\ketbra{10}_{\bm X_\xi \bm X_{\xi'}}$ and cannot erase the information on the direction of the tokens, which is required for nonzero coherence. Therefore the quantum distribution is genuine network nonlocal \cite{Supic2022,Sekatski2023}.

\subsubsection{All states are entangled}

Next, we show that all states must be entangled. For concreteness, consider the state 
$\ket{\psi_\alpha}_{B_\alpha C_\alpha E_\alpha}= \sqrt{p_\alpha} \ket{01}_{\textbf{B}_\alpha \textbf{C}_\alpha}\!\ket{j^{01}_\alpha}_{B'_\alpha C'_\alpha E_\alpha}\!\!+ \sqrt{1-p_\alpha} 
\ket{10}_{\textbf{B}_\alpha \textbf{C}_\alpha}\ket{j^{10}_\alpha}_{B'_\alpha C'_\alpha E_\alpha}$, and trace out the "eavesdropper's" system $E_\alpha$ which remains at the source to define the state
\be
\rho^{\alpha}_{\bm B_\alpha B'_\alpha \bm C_\alpha C'_\alpha} =\text{tr}_{E_\alpha}\ketbra{\psi_\alpha}_{B_\alpha C_\alpha E_\alpha}.
\ee
 Following \cite{Sekatski2023} note that all decomposition of the state into pure states 
$\rho^\alpha = \sum_k p_k \ketbra{\psi_k}$ 
involve pure states of the form
\be
\ket{\psi_k}= c_k \ket{01}_{\textbf{B}_\alpha \textbf{C}_\alpha}\!\ket{\xi_k}_{B'_\alpha C'_\alpha} +s_k \ket{10}_{\textbf{B}_\alpha \textbf{C}_\alpha}\!\ket{\zeta_k}_{B'_\alpha C'_\alpha} 
\ee
with real positive $c_k^2+s_k^2=1$.  This decomposition has the average entropy of entanglement given by 
\be
S = \sum_k p_k h(c^2_k),
\ee
where $h$ is the binary entropy. Furthermore, from the form of the state $\rho^\alpha$,
we know that $\sum_k p_k c_k^2 = p_\alpha$
and $\ketbra{01} \rho^\alpha \ketbra{10} = \sqrt{p_\alpha(1-p_\alpha)} \tr_{E_\xi} \ketbra{01,j^{01}_\alpha}{10,j^{10}_\alpha}_{B_\alpha C_\alpha E_\alpha}$, which leads to 
\begin{align}
R \leq \|\tr_{E_\xi} \ketbra{01,j^{01}_\alpha}{10,j^{10}_\alpha}_{B_\alpha C_\alpha E_\alpha} \|_1 &= \left\|\frac{\ketbra{01}{10}_{\textbf{B}_\alpha \textbf{C}_\alpha}\otimes \sum_k p_k c_k s_k \ketbra{\xi_k}{\zeta_k} _{B'_\alpha C'_\alpha}}{\sqrt{p_\alpha(1-p_\alpha)}}\right \|_1 \\
& = \left\|\frac{\sum_k p_k c_k s_k \ketbra{\xi_k}{\zeta_k} _{B'_\alpha C'_\alpha}}{\sqrt{p_\alpha(1-p_\alpha)}}\right \|_1 \leq \frac{\sum_k p_k c_k s_k}{\sqrt{p_\alpha(1-p_\alpha)}}
\end{align}

The entanglement of formation of the state $\rho^\alpha$ can thus be bounded by the following optimization
\begin{align}
\cE_F(\rho^\alpha) \geq  \quad \min \quad &\sum_k p_k h(c^2_k) \\
\text{such that} \quad & \sum_k p_k c_k^2 = p_\alpha \label{eq: BBB}\\
&  \sum_k p_k c_k s_k \geq  R\sqrt{p_\alpha(1-p_\alpha)}
\\&\sum_k p_k =1 \\
&c_k^2+s_k^2=1
\end{align}

For completeness let us give a simple solution of this minimization problem (which  can also be found in appendix E of \cite{Sekatski2023}). Relaxing the constraint~\eqref{eq: BBB}, and defining a new variable $x_k = c_k s_k = c_k \sqrt{1-c_k^2} \in [0,1]$ we get 
\begin{align}
\cE_F(\rho^\alpha) \geq  \quad \min \quad &\sum_k p_k h \left(\frac{1 - \sqrt{1-4x_k^2}}{2} \right) \\
\text{such that} \quad &  \sum_k p_k x_k \geq  R\sqrt{p_\alpha(1-p_\alpha)}
\\&\sum_k p_k =1
\end{align}
One can verify that the function $h \left(\frac{1 - \sqrt{1-4x^2}}{2} \right)$ is convex and increasing for $x \in [0,1]$, therefore
\begin{align}
    \sum_k p_k h \left(\frac{1 - \sqrt{1-4x_k^2}}{2} \right) \geq h \left( \frac{1 - \sqrt{1-4 {(\sum_k p_k x_k)}^2}}{2} \right) \geq h\left( \frac{1-\sqrt{1-4 R^2\,  p_\alpha(1-p_\alpha)}}{2}\right)
\end{align}
implying
\be\label{eq: EF final}
\cE_F(\rho^\alpha) \geq h\left( \frac{1-\sqrt{1-4 R^2\,  p_\alpha(1-p_\alpha)}}{2}\right).
\ee

\subsubsection{Certified randomness}

Next, we bound the randomness of the outputs with respect to an eavesdropper that can have access to the purifying systems $E=E_\alpha, E_\beta, E_\gamma$. For concretes we focus on the output $a_i$, and only look at the randomness of the coarse-grained output $a$ which is determined by the token degrees of freedom received by the party $A$. Furthermore, we also coarse grain the outputs $a=0$ and $a=2$, so that at the end we are left with binary output $\tilde a$ which equals to the parity of $a$. The state of the corresponding classical register denoted $\bar A$ is then
\begin{align}
\ketbra{00}_{\bm{A}_\beta \bm{A}_\gamma} \quad &\text{or} \quad \ketbra{11}_{\bm{A}_\beta \bm{A}_\gamma} \mapsto  \ketbra{0}_{\bar A} \\
\ketbra{01}_{\bm{A}_\beta \bm{A}_\gamma} \quad &\text{or} \quad \ketbra{10}_{\bm{A}_\beta \bm{A}_\gamma} \mapsto  \ketbra{1}_{\bar A}.
\end{align}
And the classical quantum state of the register $\bar A$ and eavesdropper's systems read
\begin{align}
\rho_{\bar{A} E}  =\rho_{\bar{A} E_\gamma E_\beta}  = &\ketbra{0}_{\bar A}\otimes \left( (1-p_{\beta})p_{\gamma}\,  \rho_{E_\gamma}^{01}\otimes \rho_{E_\beta}^{10} + p_{\beta}(1-p_{\gamma}) \, \rho_{E_\gamma}^{10}\otimes \rho_{E_\beta}^{01} \right) \\
+ &\ketbra{1}_{\bar A}\otimes\left( p_{\beta} p_{\gamma} \rho_{E_\gamma}^{01}\otimes \rho_{E_\beta}^{01} + (1-p_{\beta})(1-p_{\gamma})\,  \rho_{E_\gamma}^{10}\otimes \rho_{E_\beta}^{10}\right)
\end{align}
whereas the system $E_\alpha$ is independent of party A. In this setting, the randomness of $\bar{A}$ against Eve, is lower bounded by the following Theorem.

\begin{theorem}
   Consider the quantum distribution $P^\triangle_Q(a,b,c)$ fulfilling the requirement of Theorem \ref{th: Quantum LP}, with the coherence satisfying $\sum_{ijk} (-1)^{s_{ijk}} r_{ijk} \geq R$, in this setting the following bounds holds on the randomness of the parity output $\bar a =\begin{cases} 0 & \tilde a=0,1\\1 & \tilde a=1\end{cases}$, for $p_\beta, p_\gamma = p$. Where $h(x)$ is the entropy of the binary distribution $(x, 1-x)$ and $H$ is the entropy of a four-output distribution.

    \begin{itemize}
        \item $H_\text{min}(\bar A|E)\geq -\log_2\left( \frac{1}{2} + \frac{1}{2}\left(|2p-1| +\sqrt{1-R}\right)^2\right)$

        \item For $p=1/2$, we can get a better bound:  $H(\bar A|E)\geq 1 + h\left(\frac{1+R}{2}\right) - H\left(
        \frac{(1+\sqrt{R})^2}{4},
        \frac{(1-\sqrt{R})^2}{4},
        \frac{1-R}{4},
        \frac{1-R}{4}\right)$
    \end{itemize}
    where $\bar A$ is the register containing $\bar a$.
\end{theorem}

\begin{proof}
Our goal is to lower bound $H_\text{min}(\bar A|E)$ for this state, in other words the probability $p_g(\bar A|E)= 2^{-H_\text{min}(\bar A|E)} $ that the eavesdropper guesses the value of $\bar A$. To this end let us rewrite the state as 
\begin{align}
\rho_{\bar{A} E} &= Q \ketbra{0}_{\bar A}\otimes \rho_E^{\bar 0}
+(1-Q) \ketbra{1}_{\bar A}\otimes \rho_E^{\bar 1} \qquad \text{with} \qquad
Q = (1-p_\beta) p_\gamma + p_\beta (1-p_\gamma), \\
\rho_E^{\bar 0} &= \frac{(1-p_{\beta})p_{\gamma}\,  \rho_{E_\gamma}^{01}\otimes \rho_{E_\beta}^{10} + p_{\beta}(1-p_{\gamma}) \, \rho_{E_\gamma}^{10}\otimes \rho_{E_\beta}^{01}}{Q}, \quad
\rho_E^{\bar 1} = \frac{   p_\beta p_\gamma \rho_{E_\gamma}^{01}\otimes \rho_{E_\beta}^{01} + (1-p_{\beta}) (1-p_{\gamma})\,  \rho_{E_\gamma}^{10}\otimes \rho_{E_\beta}^{10}}{1-Q} .
\end{align}
For this state the guessing probability is known to be
\begin{align}
p_g(\bar A |E) &= \frac{1}{2}\left(1+ \|Q\, \rho^{\bar0}_E-(1-Q)\rho^{\bar 1}_{E}\|_1\right) 
= \frac{1}{2}\left(1+  \|\big( p_{\gamma}\rho_{E_\gamma}^{01}-(1-p_\gamma )\rho_{E_\gamma}^{10}\big)\otimes \big( (1-p_\beta) \rho_{E_\beta}^{10} -p_{\beta}  \rho_{E_\beta}^{01}\big)\|_1\right)\\
&= \frac{1}{2}\left(1+  \| p_{\gamma}\rho_{E_\gamma}^{01}-(1-p_\gamma) \rho_{E_\gamma}^{10}\|_1 \cdot \| (1-p_\beta) \rho_{E_\beta}^{10} - p_{\beta}  \rho_{E_\beta}^{01}\|_1\right)
\end{align}
Now, for any two states $\rho ,\sigma$ one has
\be
\|p \rho - (1-p)\sigma\|_1 =
\|\frac{1}{2}( \rho - \sigma) + (p-\frac{1}{2})(\rho+\delta))\|_1 \leq \frac{1}{2}\|\rho-\sigma\|_1 +\left|2p-1\right| \leq \sqrt{1-F^2(\rho,\sigma)} +\left|2p-1\right|,
\ee
where we used the fact that the trace distance $D(\rho,\sigma)= \frac{1}{2}\|\rho-\sigma\|_1$ is upper bounded by  $\sqrt{1-F^2(\rho,\sigma)}$. The last inequality allows us to write
\begin{align}
p_g(\bar A |E) &\leq  \frac{1}{2}+  \frac{1}{2}\left(|2p_\gamma-1|+\sqrt{1-F^2(\rho_{E_\gamma}^{01},\rho_{E_\gamma}^{01})}\right)\left(|2p_\beta-1|+\sqrt{1-F^2(\rho_{E_\beta}^{10},\rho_{E_\beta}^{01})}\right)\\
&= \frac{1}{2}+  \frac{1}{2}\left(|2p_\gamma-1|+\sqrt{1-F_\gamma^2} \right)\left(|2p_\beta-1|+\sqrt{1-F_\beta^2}\right),
\end{align}
where we introduced the short notation $F_\xi= F(\rho_{E_\xi}^{10},\rho_{E_\xi}^{01})$. We know that these fidelities satisfy $F_\beta F_\gamma\geq R$ by virtue of the bound ~\eqref{eq: bound fid}. Hence to find the worst case bound $H_\text{min}(\bar A|E)\geq -\log_2(p_g^*)$ we need to maximize
\begin{align}
    p_g^* = \max_{F_\gamma,F_\beta} \quad &\frac{1}{2}+  \frac{1}{2}\left(|2p_\gamma-1|+\sqrt{1-F_\gamma^2} \right)\left(|2p_\beta-1|+\sqrt{1-F_\beta^2}\right) \\
    \text{such that}&\quad F_\gamma,F_\beta \leq 1 \\
    &\quad F_\beta F_\gamma\geq R
\end{align}

For any given values of $p_\gamma,p_\beta$ and $R$ this maximization is straightforward to perform numerically.
In the particular case where $p_\gamma=p_\beta=p$ we find that the maximum is attained for $F_\gamma=F_\beta= \sqrt{R}$ leading to 
\be
H_\text{min}(\bar A|E)\geq -\log_2\left( \frac{1}{2} + \frac{1}{2}\left(|2p-1| +\sqrt{1-R}\right)^2\right).
\ee

Furthermore, in the case $p_\gamma=p_\beta=\frac{1}{2}$ that we are interested in,  this derivation gives
\be
H_\text{min}(\bar A|E)\geq -\log_2\left(1-\frac{R}{2}\right),
\ee
and improves over the bound derived in \cite{Sekatski2023}.\\

However, in the case $p_\gamma=p_\beta=\frac{1}{2}$ we can improve the bound even further. To do so let us first assume that the states $\rho_\gamma^{01}=\ketbra{\varphi_\gamma^{01}},\rho_\gamma^{10}=\ketbra{\varphi_\gamma^{10}},\rho_\beta^{01}=\ketbra{\varphi_\beta^{01}}$ and $\rho_\beta^{10}=\ketbra{\varphi_\beta^{10}}$ are pure. In this case both $H(\rho_{AE})$ and $H(\rho_E)$ can be computed analytically as follows. For the joint entropy, we find
\begin{align}
    H(\rho_{\bar AE}) &= H
    \left(\begin{array}{cc}
         \frac{1}{2}\left(\ketbra{\varphi_\gamma^{01},\varphi_\beta^{10}}+\ketbra{\varphi_\gamma^{10},\varphi_\beta^{01}} \right)&  \\
         & \frac{1}{2}\left(\ketbra{\varphi_\gamma^{01},\varphi_\beta^{01}}+\ketbra{\varphi_\gamma^{10},\varphi_\beta^{10}} \right)
    \end{array} \right) \nonumber\\
    & = 1 + \frac{1}{2} H\left(\ketbra{\varphi_\gamma^{01},\varphi_\beta^{10}}+\ketbra{\varphi_\gamma^{10},\varphi_\beta^{01}} \right) + \frac{1}{2} H \left(\ketbra{\varphi_\gamma^{01},\varphi_\beta^{01}}+\ketbra{\varphi_\gamma^{10},\varphi_\beta^{10}} \right)\nonumber \\
    & = 1 + h\left(\frac{1+F_\gamma F_\beta}{2}\right)
\end{align}
with $F_{\gamma(\beta)} =|\braket{\varphi_{\gamma(\beta)}^{01}}{\varphi_{\gamma(\beta)}^{10}}|$. While the marginal entropy reads
\begin{align}
 H(\rho_{\bar E}) &= H
        \left( \frac{\ketbra{\varphi_\gamma^{01},\varphi_\beta^{01}}+\ketbra{\varphi_\gamma^{01},\varphi_\beta^{10}} +\ketbra{\varphi_\gamma^{10},\varphi_\beta^{01}}+\ketbra{\varphi_\gamma^{10},\varphi_\beta^{10}}}{4} \right) \\
        &= H\left(
        \frac{(1+F_\gamma)(1+F_\beta)}{4},
        \frac{(1-F_\gamma)(1+F_\beta)}{4},
        \frac{(1+F_\gamma)(1-F_\beta)}{4},
        \frac{(1-F_\gamma)(1-F_\beta)}{4}\right),
\end{align}
which can be obtained by diagonalizing the state in the first line. Now let us express the conditional von Neumann entropy via the mutual information
\begin{align}
H(\bar A|E) &= H(\bar A) - I(E:A) = 1- I(E:A) \qquad \text{where}\\
I(E:A) &= H(\bar A) + H(E)- H(\bar A E)= 1 +H(\rho_E)- H(\rho_{\bar A E}) =  J(F_\gamma,F_\beta) \\
       \text{with} \quad J(F_\gamma,F_\beta) &:=H\left(
        \frac{(1+F_\gamma)(1+F_\beta)}{4},
        \frac{(1-F_\gamma)(1+F_\beta)}{4},
        \frac{(1+F_\gamma)(1-F_\beta)}{4},
        \frac{(1-F_\gamma)(1-F_\beta)}{4}\right) \\
        & - h\left(\frac{1+F_\gamma F_\beta}{2}\right).
\end{align}

At this point let us relax the assumption that the conditional states of Eve $\rho_\xi^{\bm x}$ on the systems $E_\xi$ are pure. Nevertheless, each pair of state admits purifications $\ket{\varphi_\xi^{\bm x}}_{E_\xi E_\xi'}$ on $E_\xi '\bar E_\xi'$ such that 
\be
|\braket{\varphi_{\xi}^{01}}{\varphi_{\xi}^{10}}_{E_\xi E_\xi'}| = F(\rho^{01}_{E_\xi},\rho^{10}_{E_\xi}) = F_\xi.
\ee
For the purified systems we thus find $I(EE':A) =  J(F_\gamma,F_\beta)$
But the data processing inequality also guarantees that $I(E:A)\leq I(EE':A)$. Therefore, even if the states are not pure we find that 
\begin{align}
H(\bar A|E) = 1- I(E:A) \geq 1- I(EE':A) = 1-J(F_\gamma,F_\beta).
\end{align}

To find the worst case entropy compatible with our constraint $F_\gamma F_\delta\geq R$, it remains to minimize $J(F_\gamma,F_\beta)$. By plotting the function we see that the maximum is attained at $F_\gamma=F_\beta=\sqrt{R}$, leading to the bound
\be\label{eq: bound randomness}
H(\bar A|E)\geq 1 + h\left(\frac{1+R}{2}\right) - H\left(
        \frac{(1+\sqrt{R})^2}{4},
        \frac{(1-\sqrt{R})^2}{4},
        \frac{1-R}{4},
        \frac{1-R}{4}\right).
\ee

\end{proof}

To illustrate this bound, we consider the RGB4 distribution on the triangle \cite{Renou_2019} where the sources distribute $\ket{\psi^+}$ and the parties perform the PVMs in Eq.~\eqref{eq:measurment}. For the resulting distribution, we can lower bound the randomness of the outputs  (and the entanglement of the sources) in the range $0.886 < u < 1$, depicted in Fig. \ref{fig:rgb4_plot}.  It is noteworthy that the new bound on randomness gives a four-fold improvement over the previously known bound derived in \cite{Sekatski2023} (the other bounds remain unchanged), yielding $\approx 14.2\%$ of a random bit for the optimal value of the parameter $u$.


\begin{figure}[H]
    \centering \includegraphics[width=0.6\columnwidth]{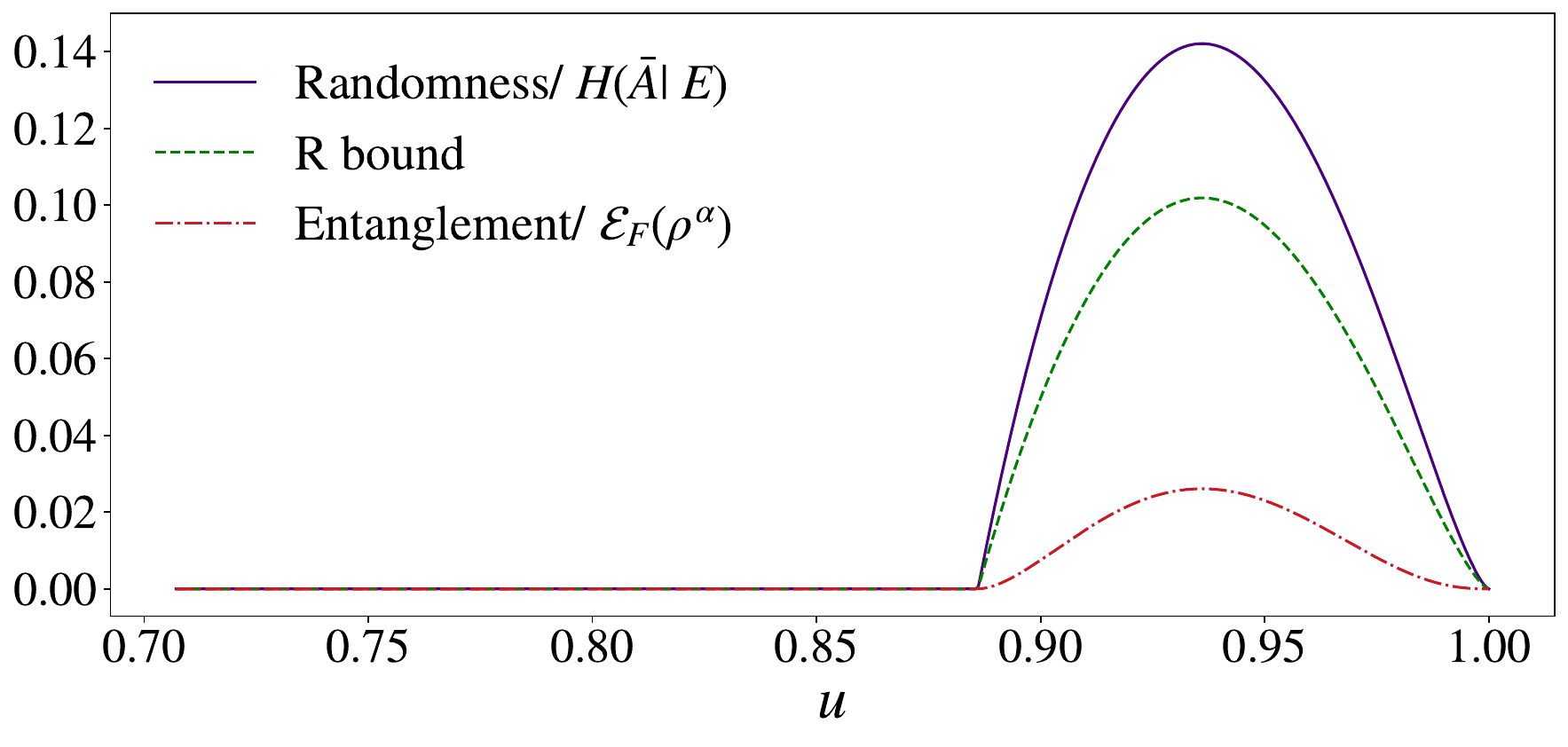}
    \caption{Lower bound on the randomness of the outputs $H(\bar A|E)$, entanglement of formation produced by the sources $\mathcal{E}_F(\rho)$ and the global coherence $R$, for the RGB4 distribution \cite{Renou_2019} on the triangle network. Recall that the distribution is obtained with the sources distributing the states $\ket{\psi^+}=\frac{1}{\sqrt 2}(\ket{01}+\ket{10})$ and the parties performing the projective measurements in Eq.~\eqref{eq:measurment}.}
    \label{fig:rgb4_plot}
\end{figure}

\section{The networks of Fig.~\ref{fig:sq_normal_extended}. Proof of Results 1 and 1'} 
\label{app:A}

In this Appendix, we give  the  full proof of the Results 1 and 1' in the main text. Specifically, we consider the quantum distribution $P_Q^\square(a_1,a_2,a_3,a_4) $ discussed in the main text, and show that it is nonlocal for the parameter range $u_{\mathrm{max}}< u<1$ with $u_{\mathrm{max}} \approx 0.841$, with respect to the networks in Fig. 1 (a), (b) and (c). Then we also derive bounds on the randomness and entanglements of the network in Fig. 1(c). As it is the strongest network, giving the most power to the eavsedropper, these results automatically hold for the networks of Fig,~1(a,b), establishing the Result 1'. Note that a bound on randomness and nonlocality in the network of Fig.~1(c), discussed in Sec.\ref{app: result1primec}, automatically implies the result 1. However, for pedagogical reasons, we first briefly give stand-alone proofs of nonlocality. For the network in Fig.~1(d), in the last section of the appendix, we provide numerical evidence that the quantum distribution is also nonlocal.

We start by recalling the quantum model leading the distribution $P_Q^\square(a_1,a_2,a_3,a_4) $ on the square network introduced in the main text. The network consists of four parties, outputting $a_1$ to $a_4$, who share bipartite quantum states distributed by four independent sources, $S_1$ to $S_4$, as in Fig. 1(a). The distribution is obtained by having each source distributing a two-qubit Bell state $\ket{\psi^+} = \frac{\ket{10}+\ket{01}}{\sqrt{2}}$. Each party receives two qubits, coming from the  independent neighbouring sources, and perform a measurement in the following basis
\begin{equation}
\label{app :measurment}
  \begin{split}
      & \ket{\bar{0}} = \ket{00}, \quad \ket{\bar{1}_0} =  u\ket{01}+v\ket{10},  \\ & \ket{\bar{1}_1} = v\ket{01}-u\ket{10},\quad \ket{\bar{2}} = \ket{11},
  \end{split}
\end{equation}
such that $u^2 + v^2 =1$ ($u,v \in \mathbb{R}$). We also use the notation 
\be
u_0 =u, \quad v_0 = v,\quad u_1 = v,\quad v_1=-u 
\ee
to write compact expressions for various probabilities.
Hence each party obtains a four-valued outcome, denoted $a_1$ to $a_4$. The resulting output probability distribution is given by 
\begin{align}\label{eq:distr_sq}
\begin{split}
P_Q^\square(a_1,a_2,a_3,a_4)  & = \big |\bra{\bar a_1,\bar a_2, \bar a_2,\bar a_4}  \, \ket{\psi^+}^{\otimes 4} \big|^2
\\
P_Q^\square(a_1,a_2,a_3,a_4)  & = \big |\bra{\phi_{a_1}} \bra{\phi_{a_2}} \bra{\phi_{a_3}} \bra{\phi_{a_4}}  \, \ket{\psi^+}^{\otimes 4} \big|^2 \ \text{for} \ \phi_{a_i} \in \{\bar{0}, \bar{1}_0, \bar{1}_1, \bar{2}\}
\end{split}
\end{align}
Where the respective Hilbert spaces are suitably ordered according to the square network configuration.

\subsection{Nonlocality proof for networks in Fig. 1(a)}

Here we simply follow the proof of \cite{Renou_2019}. The following proof is a straightforward translation of the theorem~\ref{th: semi-local} for the triangle to the square network. By TC rigidity~\ref{th:class rigid} for any classical strategy leading to the distribution in $P_Q^\square(a_1,\dots,a_4)$ on the square network we can assign a token variable $t_i \in \{0,1\}$ to each source. In particular, considering the case where all parties output either $a_j=1$, it necessitates that all the tokens must be transmitted in the clockwise $\bm t =(t_1, t_2,t_3, t_4)= (1,1,1,1)$ (that we will denote $\bm t =\cw$) or anti-clockwise  $\bm t=(0,0,0,0)$ (denoted $\bm t= \acw$) direction. Thereby for any local model, we can define the following joint probability distribution for $\bm t \in \{ \circlearrowright , \circlearrowleft \}$.
\begin{equation}
   q(i,j,k,l,\bm t) = \text{Pr} \big( a_1=\bar{1}_i,\ a_2=1_j,\
   a_3=1_k,\
   a_4=\bar{1}_l,\ t \ | \ a_1,a_2,a_3,a_4 \in \{\bar{1}_0, \bar{1}_1 \} \big) 
\end{equation}
\begin{equation}
   q(i,j,k,l|\bm t) = \text{Pr} \big( a_1=\bar{1}_i,\ a_2=1_j,\
   a_3=1_k,\
   a_4=\bar{1}_l,\ t \ | \ a_1,a_2,a_3,a_4 \in \{\bar{1}_0, \bar{1}_1 \} \big) 
\end{equation}

This distribution is hidden from us but it requires to satisfy the following
\begin{equation}
\begin{split}
     &\sum_{t =  \circlearrowright , \circlearrowleft} q(i,j,k,l,t) = q(i,j,k,l) = \frac{1}{2} \big( u_i u_j u_k u_l+ v_i v_j v_k v_l \big)^2\\
     &\sum_{jkl} q(i,j,k,l,\circlearrowright) = q(i,\circlearrowright) = \frac{v_i^2}{2}, \ \ \  \sum_{jkl} q(i,j,k,l,\circlearrowleft) = q(i,\circlearrowleft) = \frac{u_i^2}{2}  \ \ \ \ \  (*)
\end{split}
\label{app:LP_square}
\end{equation}
and similar constraints to $(*)$ for $q(j, t), q(k, t)$ and $q(l, t)$.  The linear constraints specified in (\ref{app:LP_square}) define a linear program that any local model must satisfy. Hence, the LP's unfeasibility establishes the nonexistence of such a local model, affirming the distribution's nonlocality. Numerical verification of this LP leads to the following result
\be
P_Q^\square(a_1,a_2,a_3,a_4) \textit{ is nonlocal on the square network of Fig.~1(a) if } u\in(u_{\mathrm{max}}, 1), \, \textit{with } u_{\mathrm{max}}\approx 0.8408965 .\nonumber
\ee

\subsection{Nonlocality proof for networks in Fig. 1(b)}

Next, we aim to demonstrate nonlocality with respect to the network in Fig.~1(b). Note that this network, albeit involving a tripartite source, is also NDCS network, allowing us to apply the token counting rigidity result~\ref{th:class rigid}. It implies that the sources $S_1$ and $S_2$ each uniformly distribute one token to their connected parties allowing us to define $t_1$ and $t_2$ in the same way as before. In turn, the source $S_3$ allocates two tokens among the parties $A_1, A_3, A_4$ based on the probability distribution 
\be
p_{t_3^1,t_3^3,t_3^4}(0,1,1) = p_{t_3^1,t_3^3,t_3^4}(1,0,1) = p_{t_3^1,t_3^3,t_3^4}(1,1,0) = p_{t_3^1,t_3^3,t_3^4}(0,0,2) = 1/4,
\ee
where $t_3^j$ indicates the number of tokens sent from $S_3$ to the party $A_j$ and we have $\sum_{j \in \{1,3,4\}}t_3^j = 2$.
Note that this distribution mirrors that achieved through two independent uniform sources $S_3$ and $S_4$ in network (a). Let's again consider the case when all parties output $\bar{1}$, in this case $A_4$ receives its sole token from the singular connected source $S_3$, while $S_3$'s other token is uniformly allocated between $A_1$ and $A_3$, as indicated by $p_{t_3^1,t_3^3}(0,1|t_3^4 = 1) = p_{t_3^1,t_3^3}(1,0|t_3^4 = 1) = 1/2$. Consequently, in this scenario , it is clear that token transmission occurs either in a clockwise ($t,{=}\circlearrowright$) or anti-clockwise ($t,{=}\circlearrowleft$) direction (on the triangle formed by the parties $A_1,A_2,A_3$). This leads us to the same linear program (LP) as denoted in \ref{app:LP_square}, whose unfeasibility indicates the nonlocal nature of the distribution with respect to the network (b) as well. In conclusion, we find that

\be
P_Q^\square(a_1,a_2,a_3,a_4) \textit{ is nonlocal on the square network of Fig.~1(b) if } u\in(u_{\mathrm{max}}, 1), \, \textit{with } u_{\mathrm{max}}\approx 0.8408965.
\ee

\subsection{Proof of nonlocality for the network in Fig. 1(c) and of the Result 1'.}
\label{app: result1primec}

Finally, we illustrate nonlocality of $P_Q^\triangle(a_1,a_2,a_3,a_4)$ within the network in Fig.~1(c). This is a triangle network and to match the labeling let us denote the parties producing the outputs $a_1$ and $a_4$ by $B$ and $C$ respectively, and denote the outputs $a_1=b$ and $a_4= c$. Both of which belong to the alphabet $\{0,1_0,1_1,2\}$. The party producing the outputs $a_2$ and $a_3$ will be denoted $A$, and the pair of outputs is coarse-grained to $a$ as given in the next equation. The outputs $a,b,c$ are thus given by 
\be\label{eq: a2a3cg}
a = \begin{cases}
0 & (a_2,a_3) = (1,0) \text{ or } (0,1)\\
1_0 & (a_2,a_3) = (1_0,1_0)\\
1_1 & (a_2,a_3) = (1_0,1_1)\\
1_2 & (a_2,a_3) = (1_1,1_0)\\
1_3 & (a_2,a_3) = (1_1,1_1)\\
1_4 & (a_2,a_3) = (0,2)\\
1_5 & (a_2,a_3) = (2,0)\\
2 & (a_2,a_3) = (1_i,2) \text{ or } (2,1_i)
\end{cases}
\qquad b = a_1 \in\{0,1_0,1_1,2\} \qquad c= a_4 \in\{0,1_0,1_1,2\},
\ee
note that the outputs $(a_2,a_3) =(0,0)$ and $(2,2)$ are impossible.  

This relabeling and coarse-graining defines the  distribution $P_Q^\triangle(a_1,a_2,a_3,a_4) \to P_Q^\triangle (a, b, c)$, given by
\begin{align}
 P(2,1_j, 0)&=  P(0,2, 1_j)=\frac{1}{8} u_j^2  ,  \qquad \qquad  \qquad \qquad P(0,1_j,2)= P(2,0,1_j)= \frac{1}{8} v_j^2  
 \\
P(1_i,0, 2 ) &= \begin{cases}
\frac{1}{16} u_{i_0}^2 u_{i_1}^2 & i=i_0i_1  \in \{ 0,1,2,3\}\\
\frac{1}{16} & i =4 \\
0 & i =5
\end{cases},
\quad
P(1_i, 2, 0 ) = \begin{cases}
\frac{1}{16} v_{i_0}^2 v_{i_1}^2 &  i=i_0i_1  \in \{ 0,1,2,3\}\\
0 & i =4 \\
\frac{1}{16} & i =5
\end{cases}\\
&P(1_i, 1_j, 1_k) = \begin{cases}
\frac{1}{16}(u_{i_0}u_{i_1} u_j u_{k} + v_{i_0}v_{i_1} v_j v_{k})^2 &  i= i_0i_1 \in \{ 0,1,2,3\}\\
\frac{1}{16} u_j^2 u_k^2 &  \ i=4\\
\frac{1}{16} v_j^2 v_k^2 &  \ i=5
\end{cases}.
\end{align}
 It is easy to see that upon further coarse-graining the outputs to $\tilde a, \tilde b, \tilde c$ (forgetting the subscripts of $1_\ell$ in Eq.~\eqref{eq: a2a3cg}) we get the distribution $P_Q^\triangle(\tilde a,\tilde b, \tilde c)$. Hence, by virtue of Theorem~\ref{the: LP nonloc tri} we can readily disprove its locality with the LP in Eq.~\eqref{eq: feasibility locality}, which leads to the following result

\be
P_Q^\triangle(a_1,a_2,a_3,a_4) \textit{ is nonlocal on the network of Fig.~1(c) if } u\in(u_{\mathrm{max}}, 1), \, \textit{with } u_{\mathrm{max}}\approx 0.8408965.\nonumber
\ee

Next, in addition to exhibiting the nonlocality of $P_Q^\triangle(\bar a, \bar b,\bar c)$ with $\bar a$ defined in Eq.~\eqref{eq: a2a3cg}, we bound some of its quantum properties. Since it becomes token counting under coarse-graining, the quantum TC rigidity results~\ref{th:quant rigid} implies that the underlying quantum modes are of the form given in Eqs.~(\ref{eq: Pi0}-\ref{app: pf}). In particular, let us remind that the states read 
\be
 \ket{\psi_\xi}_{X_\xi Y_\xi E_\xi}= \sqrt{p_\xi} \ket{10}_{\textbf{X}_\xi \textbf{Y}_\xi}\!\ket{j^{10}_\xi}_{J_\xi}\!\!+ \sqrt{1-p_\xi} 
\ket{01}_{\textbf{X}_\xi \textbf{Y}_\xi}\ket{j^{01}_\xi}_{J_\xi}
\ee
where $p_\xi$ are the token probabilities of the honest distributions. In the case where all parties receive one token the output probabilities read 
\be
P_Q(1_i, 1_j, 1_k) = \text{Pr}(\cw) \bra{\cw } \Pi_A^{1_i}\Pi_B^{1_j}\Pi_C^{1_k}  \ket{\cw} + \text{Pr}(\acw) \bra{\acw}  \Pi_A^{1_i}\Pi_B^{1_j}\Pi_C^{1_k} \ket{\acw} +  2 \sqrt{\text{Pr}(\acw)\text{Pr}(\cw)} r_{ijk},
\ee
where $r_{ijk}=\text{Re}\bra{\acw } \Pi_A^{1_i}\Pi_B^{1_j}\Pi_C^{1_k}  \ket{\cw}$, $\text{Pr}(\cw) = p_\alpha\, p_\beta \, p_\gamma$, $\text{Pr}(\acw) = (1-p_\alpha)(1- p_\beta) (1-p_\gamma)$, $\ket{\acw}  =\ket{01,01,01}_{\textbf{B}_\alpha \textbf{C}_\alpha \textbf{C}_\beta \textbf{A}_\beta \textbf{A}_\gamma \textbf{B}_\gamma}\ket{j_\alpha^{01},j_\beta^{01},j_\gamma^{01}}_{J_\alpha J_\beta J_\gamma}$, and $
\ket{\cw}   =\ket{10,10,10}_{\textbf{B}_\alpha \textbf{C}_\alpha \textbf{C}_\beta \textbf{A}_\beta \textbf{A}_\gamma \textbf{B}_\gamma}\ket{j_\alpha^{10},j_\beta^{10},j_\gamma^{10}}_{J_\alpha J_\beta J_\gamma}.$ In our case the indices $j$ and $k$ are binary, which guarantees $r_{i1k}=- r_{i0k}=$ and $r_{ij1}= -r_{ij0}$, or simply 
\be
r_{ijk}= (-1)^{j+k} \, r_{i00}=   (-1)^{j+k} \, r_{i} \quad \text{for} \quad i =0,\dots, 5
\ee
Let us now fix $s =(+1,-1,-1,+1,-1,+1)$ and consider the  following 
sum of coherences
\be
R = \sum_{i=1}^{5} \sum_{j,k=0}^1 (-1)^{s_i} (-1)^{j+k} r_{ijk} = \sum_{i=1}^{5} \sum_{j,k=0}^1 (-1)^{s_i} r_{i} = 4 \sum_{i=1}^{5} (-1)^{s_i} r_{i}
\ee
By virtue of Theorem \ref{th: Quantum LP} it can be lower bounded with the LP in Eq.~\eqref{eq: optimization r}, expressed in terms of the variable $r_{ijk}$ and 
\be
\begin{split}
q(i,j,k|\cw) &= \text{Pr}(a=1_i,b=1_j,c=1_k| \bm t =\cw)
=  \bra{\cw } \Pi_A^{1_i}\Pi_B^{1_j}\Pi_C^{1_k}  \ket{\cw},\\
q(i,j,k|\acw) &= \text{Pr}(a=1_i,b=1_j,c=1_k| \bm t =\acw)   =
 \bra{\acw } \Pi_A^{1_i}\Pi_B^{1_j}\Pi_C^{1_k}  \ket{\acw}.
\end{split}
\ee
In addition, in this case, the non-linear constraints of Eq.~\eqref{eq: new bound r} can be easily added to the LP. This is because in our case the constraints  (\ref{th3: c3}, \ref{th3: c4}) guarantee that
\begin{align}
q(4,j, k|\cw)\leq q(i=4|\cw) = \frac{P^\triangle(1_{20},2,0)}{P^\triangle(1,2,0)} =0 \qquad
q(5,j, k|\acw)\leq q(i=5|\acw) =\frac{P^\triangle(1_{02},0,2)}{P^\triangle(1,0,2)}=0 
\end{align}
Hence for $i=4,5$ Eq.~\eqref{eq: new bound r} becomes simply
\be
|r_{4jk}| \leq \sqrt{q(4,j, k|\acw),q(4,j, k|\cw)} = 0 \qquad
|r_{5jk}| \leq \sqrt{q(5,j, k|\acw),q(5,j, k|\cw)} = 0
\ee
eliminating the variables $r_{4}$ and $r_{5}$. This leaves us with four "coherence" variable $r_0,\dots,r_3$ and the linear program in Eq.~\eqref{eq: optimization r} can be written as follows:

\begin{align}
R\geq\qquad  \min_{\substack{ q(i,j,k|\cw)\\ q(i,j,k|\acw)\\ r_0,r_1,r_2,r_3}} \quad & 4~(r_{0}-r_{1}-r_{2}+r_{3})\\
\text{such that}\quad&  q(i,j,k|\cw), q(i,j,k|\acw) \text{ are probability distributions}
\\
&\sum_i r_{i} = 0\\
& q(i,j,k|\cw) +  q(i,j,k|\acw) +  (-1)^{j+k}  \ 2r_{i} =  8 P(1_i, 1_j, 1_k)\\
&q(i|\cw) =  8 P^\triangle(1_i,2,0)   \ \  \text{same for} \  j, k\\
&q(i|\acw) = 8 P^\triangle(1_i,0,2) \ \  \text{same for} \  j, k
\end{align}


This LP gives a nonzero lower bound on $R$ for $0.841<u<1$ and as shown in the section \ref{app: randomness and entanglement} this bound implies that all the measurements are entangled, and gives quantitative bounds on the entanglement of the source and the randomness of the output which are depicted in Fig. \ref{fig:square_plot}

Finally, observe that the network of Fig.~1(c) is stronger than (a) and (b), therefore the above bounds on entanglement and randomness remain valid for these networks. Note that these bounds are derived with the minimum assumptions about the topology of the network, however, considering the weaker network in Fig.~1(c), by Result 2 and appendix \ref{app:result2} for $N=4$, we get better bounds for entanglement and randomness depicted in Fig. \ref{fig:square_plot}

\begin{figure}[H]
    \centering
    \includegraphics[width=0.49\textwidth]{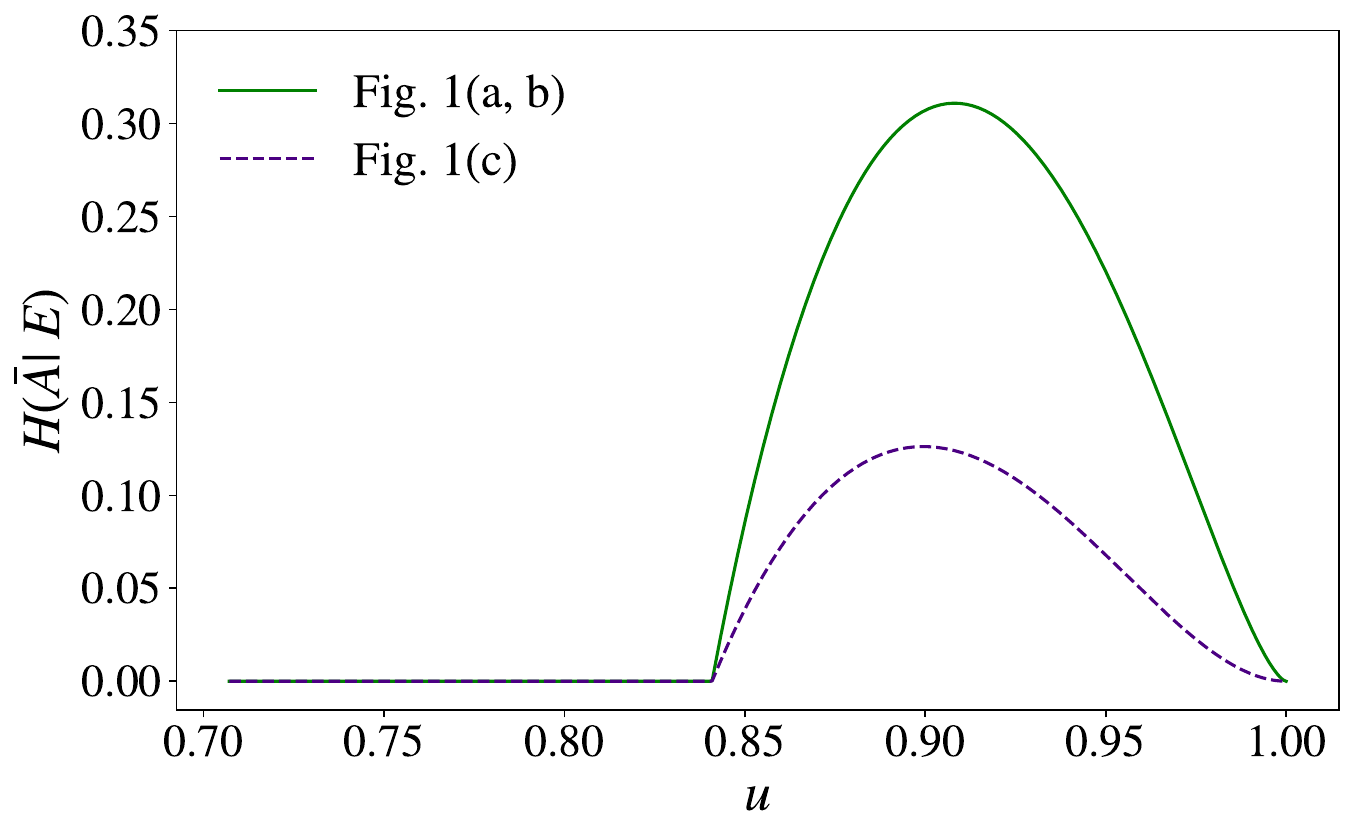}\quad
    \includegraphics[width=0.49\textwidth]{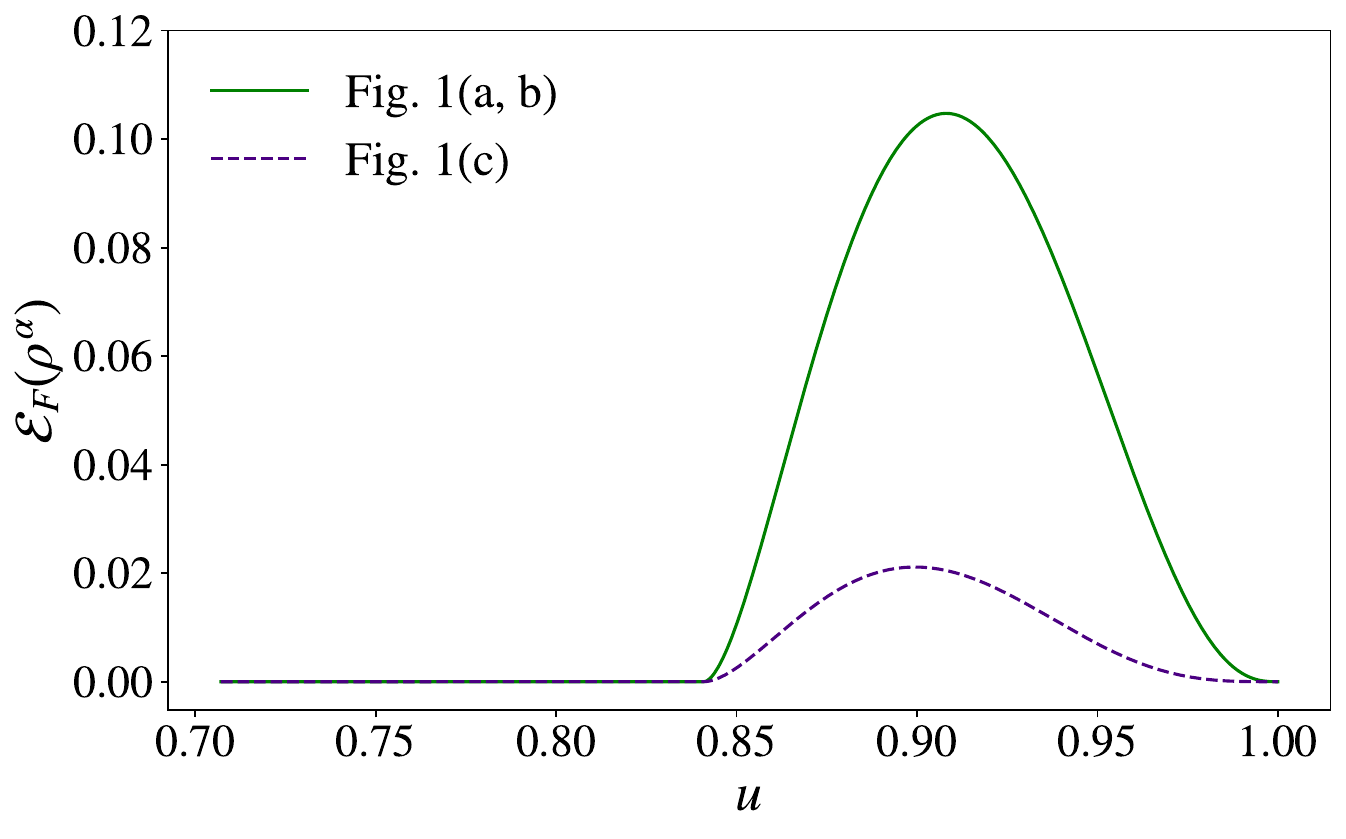}
    \caption{\textbf{Square network:} Lower bounds on the randomness of the output $a_1$ (left) and the entanglement of formation produced by the sources (right), except the source $S_2$ in Fig.~\ref{fig:sq_normal_extended}(b),  for the distribution $P_Q^\square(a_1,a_2,a_3,a_4)$ in Eq.~\eqref{eq:distr_sq} with respect to the networks in the Fig.~\ref{fig:sq_normal_extended}(a,b) (top curve, green), and Fig.~\ref{fig:sq_normal_extended}(c) (bottom curve, purple). For the three networks the distribution in nonlocal in the range $0.841 \approx u_{max}<u<1$.}%
    \label{fig:square_plot}%
\end{figure}

\subsection{Numerical analysis for network in Fig. 1 (d)}
\label{Numerical_noise_robustness}

 In this section, we investigate numerically the nonlocality of the quantum distribution $P_Q^\square(a_1,a_2,a_3,a_4) $ with respect to the various networks of Fig. 1. In particular, the focus is on network (d), for which our analytical proof does not apply. For the other networks (except (c)), we present some estimates of the noise tolerance of the quantum distribution.

We use the generative neural network method presented in Ref. \cite{krivachy_neural_2020}, referred to as LHV-Net. The algorithm aims to construct a distribution $P_{NN}$ that is local by construction, that is as close as possible to a given target distribution $P_{target}$. More precisely, the algorithm minimizes the Euclidean distance 

\begin{equation}
\begin{aligned}
    \label{distance_formula}
     d (P_{target}, P_{NN}) = \sqrt{\sum_{a,b,c,d} \left[P_{target}(a,b,c,d)-P_{NN}(a,b,c,d) \right]^2}.
\end{aligned}
\end{equation}
This method gives us an upper bound on the distance between the quantum distribution and the closest local distribution. A low distance (usually smaller than $0.001$) provides strong evidence that the target distribution is indeed local (or at least there exists a local distribution that is extremely close). On the contrary, when the algorithm cannot find a good approximation to the target distribution, we get some evidence that the latter is nonlocal.

We apply LHV-Net to a family of quantum distributions obtained by adding noise to the initial (noiseless) quantum model for the distribution $P_Q^\square(a_1,a_2,a_3,a_4) $. Specifically, we replace the pure entangled states $\ket{\psi^+}$ produced by each source by a noisy (mixed) entangled state of the form 
\begin{equation}
\begin{aligned}
    \label{noisy_distribution}
    \rho = V \ket{\psi^+}\bra{\psi^+} + (1-V) \frac{\id}{4},
\end{aligned}
\end{equation}
where $V$ denotes the visibility. The measurement performed by each party is still of the form of Eq. \eqref{eq:measurment}. Note that when $V<1$ the resulting quantum distribution does no longer satisfy the TC condition, so we can no longer apply rigidity arguments.

\begin{figure}[t]
    \centering
    \includegraphics[width=0.8\columnwidth]{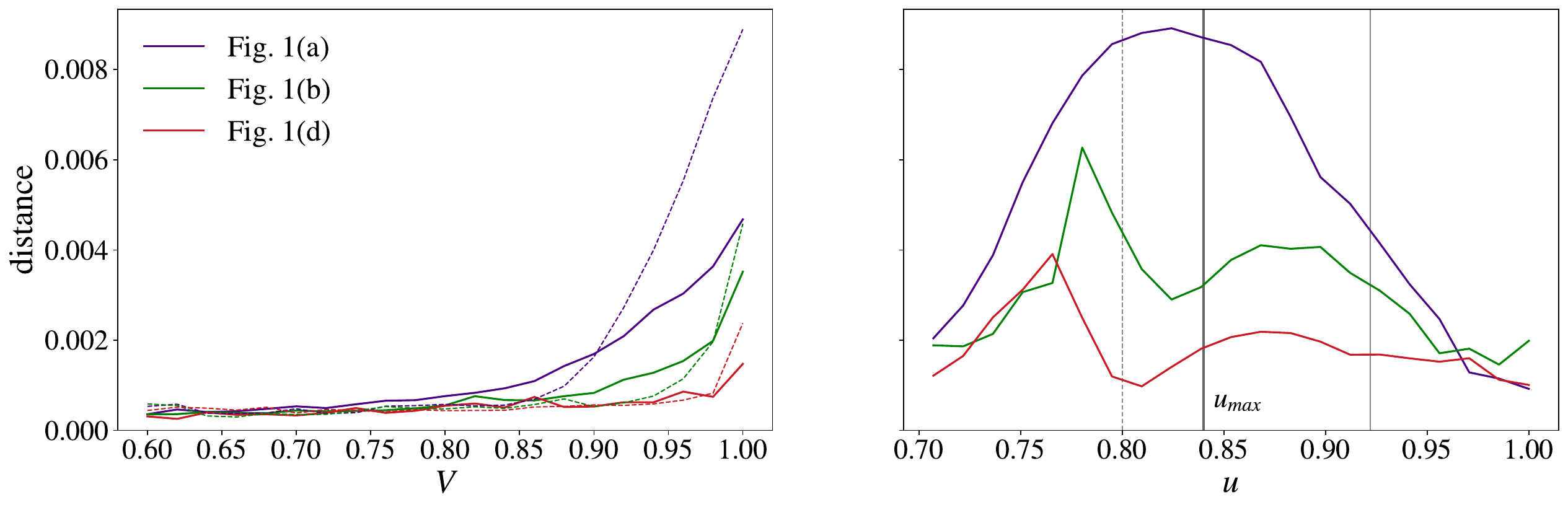}
    \caption{Results of the neural network (LHV-Net) analysis for the target quantum distribution on the square network. Local models on three different networks of Fig. (1) are considered: the square network (purple), the network in (b) (green), and the network in (d) (red). (Left) The minimal distance found by LHV-Net as a function of the amount of noise, as quantified by the visibility $V$. Two different quantum distributions are considered, for measurement parameter $u=\sqrt{0.85}\approx 0.92$ (solid lines) and $u=0.8$ (dashed lines). (Right) Minimal distance as a function of the measurement parameter $u$. The horizontal lines indicate relevant values of $u$, namely $u_{max}$ from Result 1, and the two values used in the left plot.}
    \label{fig:NN_results}
\end{figure}

In Fig.~\ref{fig:NN_results}(left) we plot the minimal distance found by LHV-Net as a function of the visibility $V$ of the quantum distribution (the target). The results indicate that, for sufficiently low noise, the quantum distribution is nonlocal. Here we consider local models based on three different networks of Fig. 1, namely the square network (a) in purple, the network (b) in green, and the network (d) in red. For networks (a) and (b), the results of LHV-Net indicate that the noiseless distribution ($V=1$) is nonlocal, as proven analytically above, and give an estimate of the noise robustness. For network (d), for which our analytic proof techniques do not work, we observe that LHV-Net predicts that the noiseless distribution is also nonlocal, though with a reduced noise robustness. More generally, we see that when considering stronger networks, the noise robustness of the quantum distribution decreases, as we expect intuitively. Note that here the analysis is conducted using two different (but fixed) values of the measurement parameter $u$.

The plot in Fig.~\ref{fig:NN_results} (right) shows the distance of the noiseless quantum distribution (with respect to the best found local distribution) for different values of the measurement parameter $u \in [ \frac{1}{\sqrt{2}},1 ] $ given in Eq.~\eqref{eq:measurment}. We observe that the quantum distribution appears to be topologically robust for essentially all values of $u$. These numerics indicate that proofs of robustness may extend to $u$ values beyond $u>u_{\text{min}}$ (where we have already proven topological robustness for Fig. 1 (a-c)), in a similar fashion to how nonlocality can also be proven beyond the $u>u_{\text{min}}$ regime~\cite{pozas2023}.

\subsection{Numerical analysis for the distribution using the elegant joint measurements}
\label{Numerical_elegant}

\begin{wrapfigure}{r}{0.5\textwidth}
\vspace{-0.8 cm}
    \begin{center}
      \includegraphics[width=0.5\columnwidth]{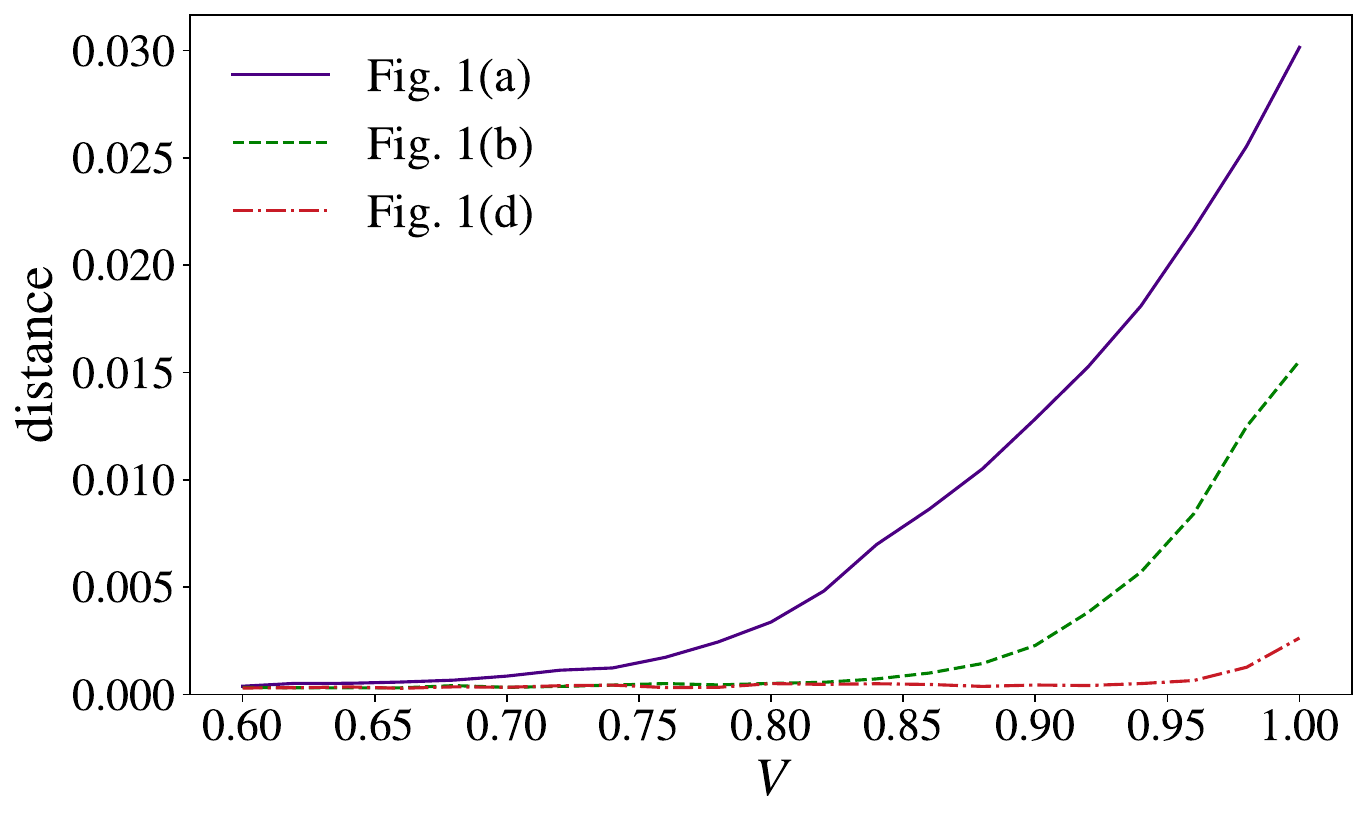}
    \end{center}
    \vspace{-0.75cm}
    \caption{Results of the neural network for the elegant distribution in the square network. The solid purple line shows the results in the original square network Fig.~\ref{fig:sq_normal_extended}(a). The dashed green line shows the results of the same distribution approached with a local hidden variable model in the network depicted in Fig.~\ref{fig:sq_normal_extended}(b). The dot-dashed red line shows the same for the strongest network Fig.~\ref{fig:sq_normal_extended}(d).}
    \label{fig:NN_results_elegant}
\end{wrapfigure}

We can consider the distribution obtained by sharing the state $\ket{\psi^-}$ and performing the elegant joint measurements introduced in Ref.~\cite{gisin2019entanglement}, which is conjectured to be nonlocal based on strong numerical evidence in the triangle network. In the square network, the distribution produced when all parties perform the elegant joint measurement has similar properties as in the triangle network, e.g. the invariance under cyclic permutation of the parties and the invariance under exchange of outputs.

Using the same numerical method as in the previous section (LHV-Net), we observe in Fig.~\ref{fig:NN_results_elegant} that the elegant distribution in the square network appears to be nonlocal (purple curve). Moreover, we see that the distribution remains nonlocal even when we consider local models on stronger networks; note that the color code is the same as in the previous subsection. Hence this represents another example of a quantum distribution that exhibits topologically robust nonlocality. Finally, we note that the distances (as found by LHV-Net) are larger than for the previous example based on token counting.


\section{The networks of Fig.~\ref{fig:general}. Proof of Result 2 and 3}

 \label{app: result 2 and 3}
\subsection{$N$-party Ring network}

The ring network consists of $N$ parties $A_1\dots A_N$ connected by bipartite sources $S_1\dots S_N$. Each source $S_k$ prepares a system $R_k$ sent to the party $A_k$ as well as a system $L_{k+1}$ sent to the party $A_{k+1}$ (with periodic boundary conditions). Hence each party $A_k$ receives $L_k R_k$. 

Let us first consider the honest quantum model over the ring network. As before the sources $S_k$ prepare Bell states $\ket{\psi^+}_{R_kL_{k+1}}$, and each party $A_i$ receiving two qubits performs the two-qubit projective measurement given in Eq. \eqref{eq:measurment} resulting in a four-valued outcome $a_i \in \{0,1_0, 1_1,2\}$. The resulting quantum distribution reads
\be \label{eq: P ring}
P_Q^{\bigcirc}(a_1,...,a_N) = \left| \bra{\bar a_1,\dots \bar a_N} \, \ket {\psi^+}^{\otimes N}\right|^2,
\ee 
and depends on the parameter $u$ describing the local measurements.

It becomes token counting when the outcomes are coarse-grained to $\tilde a_i \in\{0,1,2\}$ (by merging $\{1_0 , 1_1\} \mapsto 1$). Hence this distribution features classical rigidity formalized in Theorem~\ref{th:class rigid}, in other words, there exists a unique local model that can reproduce it. This property is again the key element to demonstrate the nonlocality of the quantum distribution $P_Q^{\bigcirc}$ over the $N$-party ring network for a well-chosen range of the measurement parameter $u$. With its help the nonlocality of the distribution $P_Q^{\bigcirc}(a_1,...,a_N)$ with respect to the honest ring network was shown in \cite{renou2022network} in the limit $u\to 1$ and for $N=4k$ and $N=2k+1$.

We now prove the nonlocality of the distribution with respect to stronger networks in Fig.~\ref{fig:general}(b,c). Remarkably for the network (b$_2$) we can prove nonlocality (also randomness and entanglement) for any $N\geq 3$ and any $u>0.84$. This automatically implies that the distribution is also nonlocal in the same parameter regime with respect to the weaker ring network, strengthening the result of~\cite{renou2022network}.

\subsection{Result 2: untrusted common source}\label{app:result2}

Let us only assume that we only trust a local part of the network depicted in the Fig.~\ref{fig:general}(b$_1$), namely that the parties $A_1$, $A_2$ and $A_3$ are separate, and $A_2$ is only connected to $A_1$ and $A_3$, and this by means of two independent bipartite sources $S_1$ and $S_2$. here we prove the following Result 2 of the main text.

\begin{wrapfigure}{r}{0.5\textwidth}
    \begin{center}
    \includegraphics[width=0.4\columnwidth]{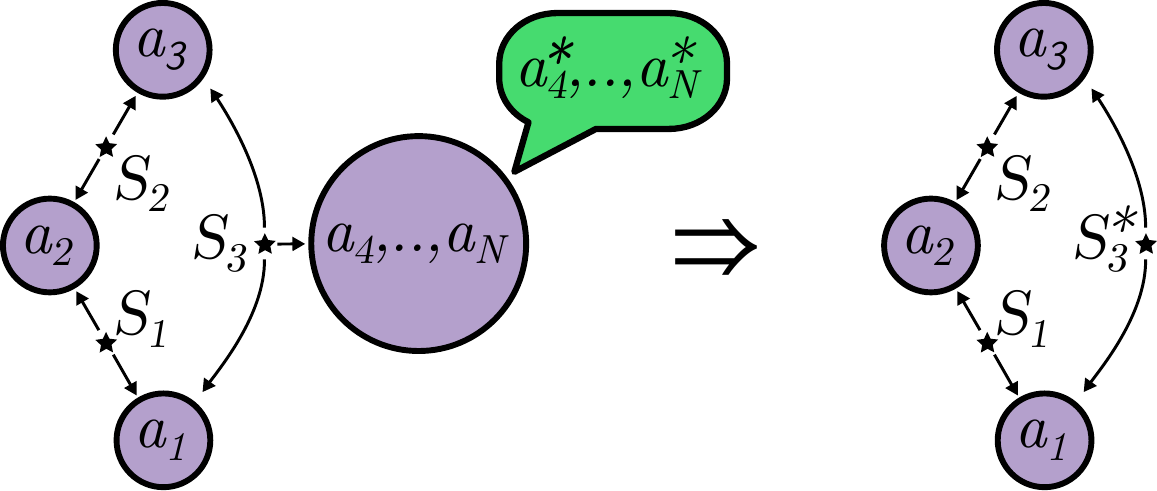}
    \end{center}
    \caption{Conditioning on the outputs $a_4^*,\dots a_N^*$ produced by the fourth party in the left network, defines a triangle network where the third source $S_3^*$ is induced by conditioning $S_3$ on the these outputs which it determines.}
    \label{fig: conditional triangle}
\end{wrapfigure}

To prove the result for all networks compatible with the assumption, we assume the strongest one, deputed in Fig.~\ref{fig:general}(b$_2$). Here the bipartite sources $S_1$ and $S_2$ connect the parties $A_1$ with $A_2$ and $A_2$ with $A_3$, and a source $S_3$ connects all the parties except $A_2$.  The idea of the proof is pretty simple and sketched in Fig.~\ref{fig: conditional triangle}. We chose a particular sequence of outputs $(a_4^*,\dots,a_{N}^*)$ produced by the corresponding parties. We then consider the \emph{conditional} distribution 
\be
P_Q^\triangle(a_1,a_2,a_3)= P^\bigcirc_Q(a_1,a_2,a_3|a_4^*,\dots,a_{N}^*).
\ee
This is a distribution on the triangle network, where the source $S_3$ is replaced by a bipartite source $S_3^*$ which prepares the systems sent to $A_1$ and $A_3$ conditionally on the outputs $a_4^*,\dots,a_{N}^*$ produced by the remaining parties connected to $S_3$.


In particular,  showing that $P_Q^\triangle(a_1,a_2,a_3)$ is nonlocal on the triangle network implies that $P_\bigcirc(a_1,\dots,a_n)$ is nonlocal with respect to the network in Fig~\ref{fig:general}(b$_2$).  Furthermore, with the results of section~\ref{app: randomness and entanglement} we can bound the entanglement produced by the sources $S_1,S_2$ or the randomness of $a_2$ based on the properties of $P_Q^\triangle(a_1,a_2,a_3)$. But, these sources and parties are independent of the post-selection performed on the source $S_3$, hence the bound on entanglement and randomness are valid for any quantum model underlying the distribution $P^\bigcirc_Q(a_1,\dots,a_N)$ before the post-selection.\\

To derive the best bounds it remains to find the best outcomes $a_4^*,\dots a_N^*$ on which to post-select. The simplest way to do so is to compute the quantum state $\ket{\psi_{\bm{a}^*}}_{R_3 L_1}$ which the source $S_3$ sends to $A_3$ and $A_1$ conditional on $\bm{a}^*=(a_4^*,\dots, a_N^*)$.  Because the outcomes $a_i=0$ and $a_i=2$ correspond to entanglement breaking measurements, it is only interesting to post-select on the outcomes $(a_4^*,\dots, a_N^*) = (1_{i_4},\dots, 1_{i_{N}})$ with $i_k \in{0,1}$. Given that in the original honest ring network, all the sources prepare $\ket{\psi^+}$ it is easy to see that the conditional state satisfies

\be
\ket{\psi_{\bm{ a}^*}}_{R_3 L_1} \propto 
\bra{1_{i_4}, \dots 1_{i_N}}_{A_4,\dots,A_{N}} \ket{\psi^+}_{R_3 L_4}\dots \ket{\psi^+}_{R_N L_1}
\ee
where $\ket{1_i} = u_i \ket{01} + v_i \ket{10}$, with $u_0=u, v_0 = v, u_1=v, v_1=-u$. Since $\ket{1_i}$ only has support on the subspace spanned by $\ket{01}$ and $\ket{10}$ it is easy to see ttat 
 \begin{equation}
\ket{\psi_{\bm{ a}^*}}_{R_3 L_1} \propto\ v_{i_4} \dots v_{i_N} \ket{01} + u_{i_4} \dots u_{i_N} \ket{10} = v_0^m v_1^{N-3-m} \ket{01} + u_0^m u_1^{N-3-m} \ket{10}
\end{equation}
where $m$ is the number of $1_0$ outcomes. We call $\chi_m$ all the combinations of outcomes resulting in the same state, i.e. where $1_0$ appears $m$ times and $1_1$ $N-3-m$ times. Given the conditional state
\be
\ket{\psi'_m}_{R_3 L_1} =  \frac{v_0^m v_1^{N-3-m} \ket{01} + u_0^m u_1^{N-3-m} \ket{10}}{\sqrt{(v_0^m v_1^{N-3-m})^2+ (u_0^m u_1^{N-3-m})^2}}
\ee
prepared with the source $S_3^*$, it is straightforward to compute the distribution on the triangle network $P_{Q}^\triangle(a_1,a_2,a_3)=P_Q^{\bigcirc}(a_1,a_2,a_3|\chi_m)$. We can analyze this distribution with the tools presented in Appendix~ \ref{app:sec_triangle}. In particular, with the help of the linear program given in Eq.~\eqref{eq: optimization r}, we can  lower bound the coherence
\be\label{eq : app R}
R = \sum_{i,j,k=0}^1 (-1)^{s_{ijk}} r_{ijk}
\ee
present in any quantum model underlying the distributions, defined in Eq.~\eqref{app: pf}. To bound the coherence $R$ note that all the indices $i,j,k$ here are binary. Following the Remark 1 in Appendix~\ref{app:sec_triangle}  we have $r_{ijk} = (-1)^{i+j+k} r$ and there is only a unique variable $r \geq 0$ to minimize, in this case, selecting $s_{i,j,k} = (-1)^{i+j+k}$ we get $R = 8 r$. We find that the optimal outputs $m$ for post-selection (maximizing $R$) depend on the size $N$ of the original network. The optimal states and the corresponding bound on $R$ are given below.

\begin{itemize}
    \item if $N = 4k$, then we consider $m=2k-1$ therefore we have $\ket{\phi} = v\ket{01} + u \ket{10}$
    
    \item if $N = 4k+1$, then we consider $m=2k$ therefore we have $\ket{\phi} = \frac{v^2\ket{01} + u^2 \ket{10}}{\sqrt{v^4+u^4}}$

    \item if $N = 4k+2$, then we consider $m=2k+1$ therefore we have $\ket{\phi} = \frac{v^3\ket{01} + u^3 \ket{10}}{\sqrt{v^6+u^6}}$

    \item if $N = 4k+3$, then we consider $m=2k-1$ therefore we have $\ket{\phi} = \frac{-u^2\ket{01} + v^2 \ket{10}}{\sqrt{v^4+u^4}}$
\end{itemize}

In turn from the bound on the coherence $R$ we can obtain a bound on the randomness of the output $a_2$ with the help of the formula~\eqref{eq: bound randomness}. Here we quantify the randomness with the von Neumann entropy of the output $H(a_2|E)$ conditional on a quantum eavesdropper holding the purification of all the states prepared by the sources. Remarkably, since the output $a_2$ does not depend on the source $S_3$, its randomness can not depend on the outputs $a_4,\dots,a_N$. That is, it is not affected by postselection and holds for the full distribution $P^\bigcirc_Q(a_1,\dots,a_N)$ on the notwork Fig.~\ref{fig:general}(c$_1$).

Similarly, the entanglement of formation produced by the sources $S_1$ and $S_2$ can be bounded from $R$ with the help of Eq.~\eqref{eq: EF final}.

\begin{figure}[H]
    \centering
    \includegraphics[width=0.48\textwidth]{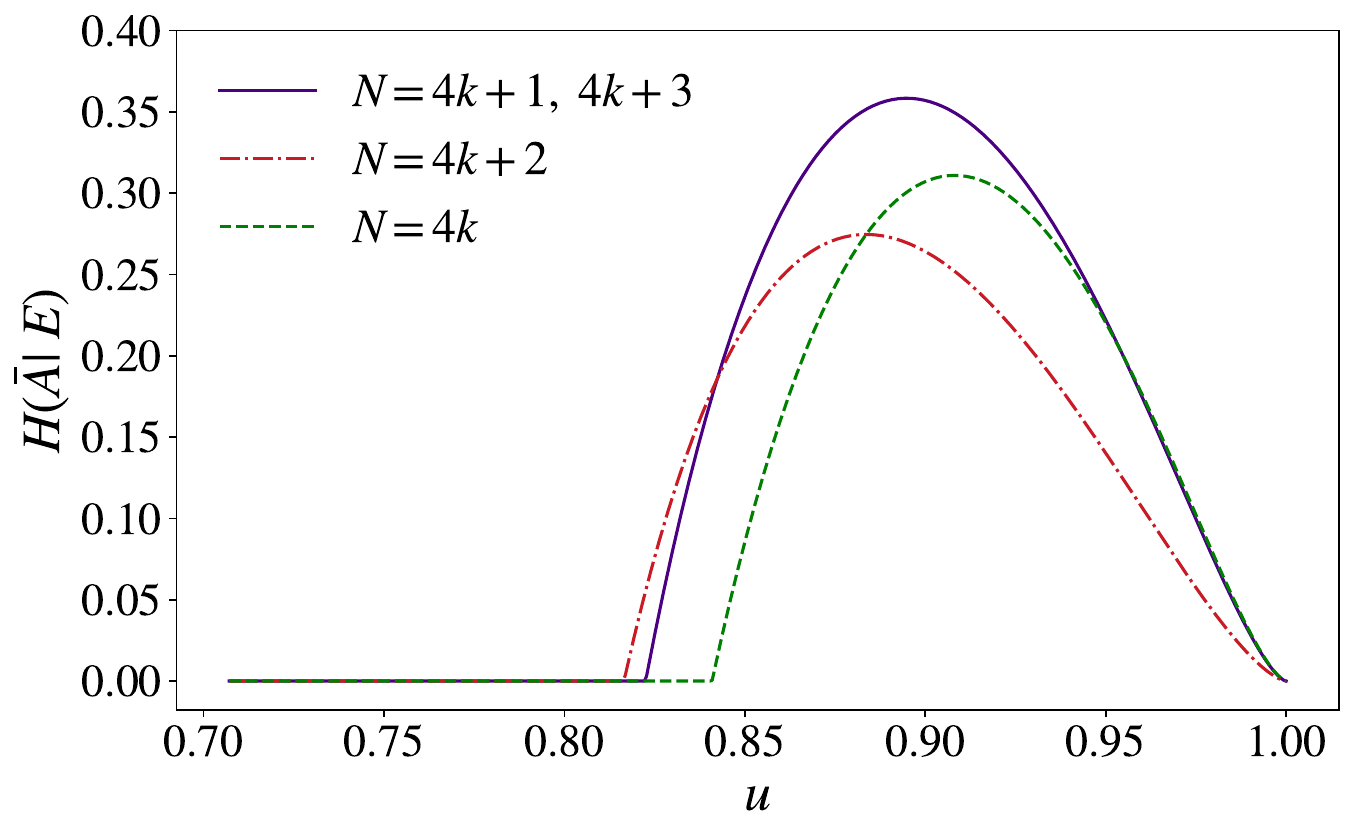} 
    \quad
    \includegraphics[width=0.48 \textwidth]{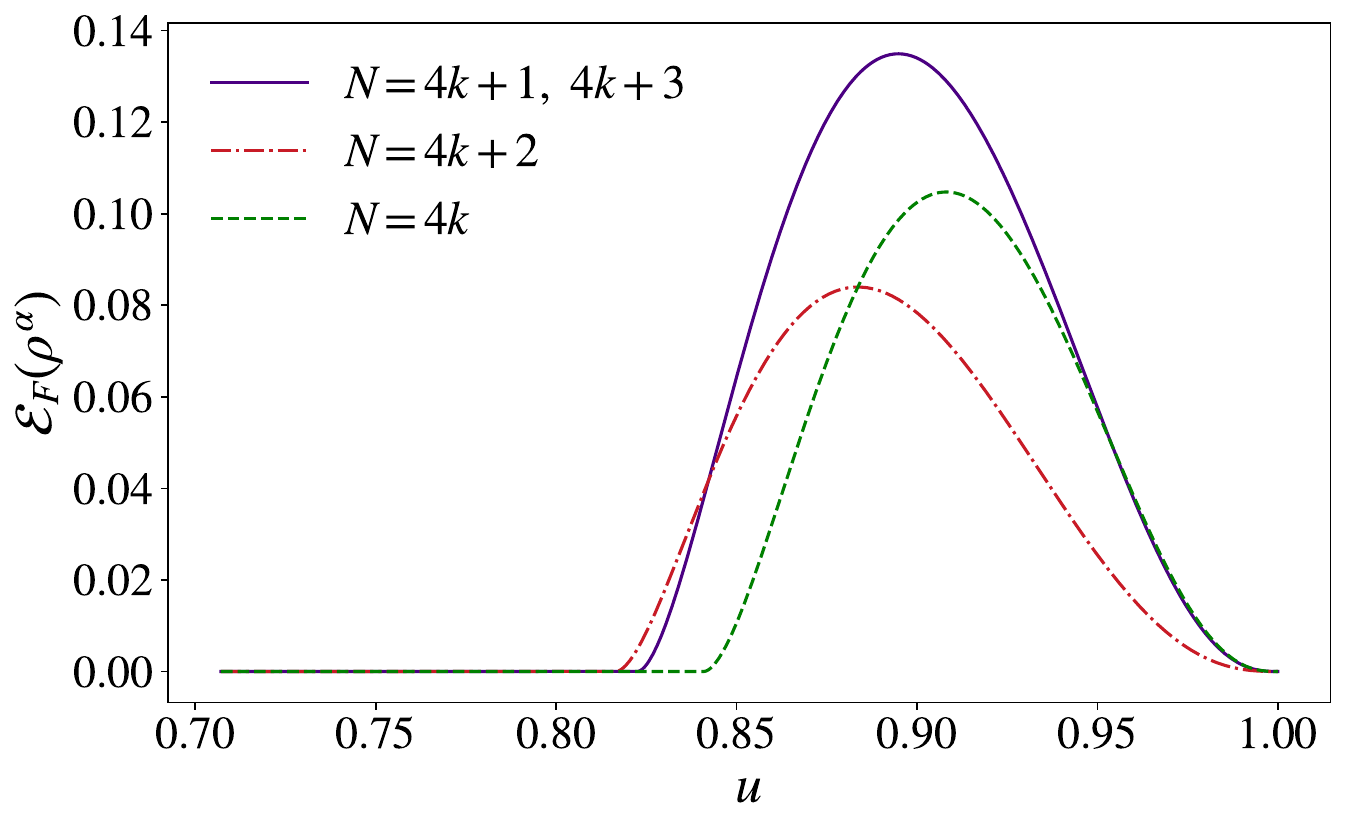} 
    \caption{Lower bounds on the randomness of the output $a_2$ (left) and the entanglement of formation produced by the sources $S_1$ and $S_2$ (right) for the distribution $P_Q^\bigcirc(a_1,\dots,a_N)$ in Eq.~\eqref{eq: P ring} with respect to the network in the Fig.~\ref{fig:general}(b$_2$) and all the networks compatible with Fig.~\ref{fig:general}(b$_1$), including the honest ring network in Fig.~\ref{fig:general}(a). The bounds differ depending on the size of the original ring network on which the quantum strategy is defined $N=4k$ (green), $N=2k+1$ (purple), and $N=4k+2$ (red). In all of the cases we find that the distribution is nonlocal for $u_\text{min}(N)<u<1$, with $ u_\text{min}(N)=0.841$ for $ N=4k$, $ u_\text{min}(N)=0.817$ for $ N=4k+2$ and $ u_\text{min}(N)=0.823$ for $ N=2k+1$. for $N=4k+2$ we prove nonlocality in the maximal range $0.817<u<1$. The bounds for the network (b) also hold for the honest square (a).}%
    \label{fig:example}%
\end{figure}

\subsection{ Results 3 : merging to a triangle}

The strongest network compatible with the assumption of Result 3, is the one where the parties $A_3, \dots A_N$  come together and act as a single entity. see Fig.~\ref{fig:general}(c$_2$) and \ref{fig: result 3}. 

\begin{wrapfigure}{r}{0.5\textwidth}
 \centering
\begin{minipage}[b]{0.05\textwidth}
    \be\nonumber
    \left. 
    \begin{array}{r}
    0 \\ 1_{0} \\ 1_{1}\\ 2
    \end{array} \right\}
    \ee
\end{minipage}
\hspace{0.3 cm}
\begin{minipage}{0.4\textwidth}
    \includegraphics[height=0.35\columnwidth]{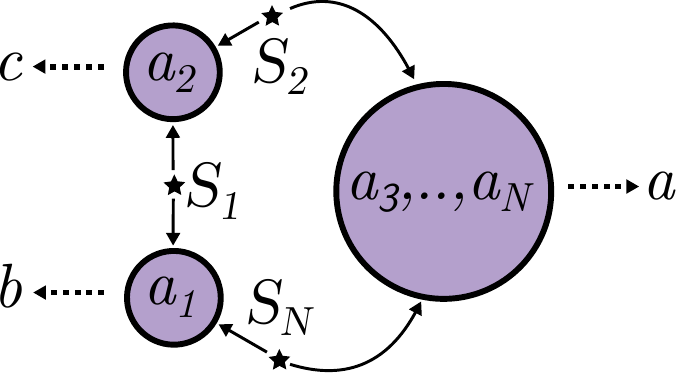} 
\end{minipage}
\hspace{-2.5 cm}
\begin{minipage}[b]{0.1\textwidth}
    \be\nonumber
    \left\{\begin{array}{l}
    0 \\ 1_{i_3,\dots,i_N} \\ 1_{20} \\ 1_{02}\\2
    \end{array}
    \right.
    \ee
\end{minipage}
    \caption{The network of Fig.~\ref{fig:general}(c$_2$) is a triangle network. For its analysis we merge the outputs $a_3,\dots, a_N$, produced by the same party, to define an output $a$ of cardinality $4+2^{N-2}$ in Eq.~\eqref{eq: merging in c}. The other outputs $a_1$ and $a_2$ are relabeled to $b$ and $c$ to match the triangle notation. The resulting distribution $P_Q^\triangle(a,b,c)$ is token counting upon coarse-graining to $\tilde a, \tilde b, \tilde c\in\{0,1,2\}$.} 
   \label{fig: result 3}
\end{wrapfigure}

This is a triangle network where the outputs $(a_3,\dots,a_N)$ are now produced by a single party that we denote $A$, the remaining two parties will be denoted $B$ and $C$, producing the respective outcomes $b=a_1$ $c=a_2$. The outcome of $A$ will be denoted $a$, it is simpler to first define its coarse-grained value
\be
\tilde a = \sum_{i=3}^N \tilde a_i - N+3.
\ee
Note that since the distribution $P_Q^\bigcirc(\tilde a_1,\dots,\tilde a_N)$ is token counting it satisfies $\sum_{i=3}^N \tilde a_i= N -\tilde a_1-\tilde a_2$ with $\tilde a_1+\tilde a_2\in \{1,2,3\}$. Hence the value of $\tilde a\in\{0,1,2\}$ is ternary.
Next, consider all the outputs $(a_3,\dots,a_N)$ leading to $\tilde a=1$. First, there is the outputs with $(a_3,\dots ,a_N) = (1_{i_3},\dots, 1_{i_N})$, which  we denote with $a=1_{i_3,\dots,i_N}$. Second there are also the outputs with $\tilde a=1$ but $(a_3,\dots ,a_N) \neq (1_{i_3},\dots, 1_{i_N})$, meaning that there is at least one output $\tilde a_i=2$ or $\tilde a_j=0$ for $3\leq i,j\leq N$. These will be coarse-grained in two groups labeled $a=1_{20}$ and $a=1_{02}$. Concretely, in this case, we find for the first output in $\{\tilde a_3,\dots, a_N\}$ which is not equal to $1$, if it is $2$ we set $a=1_{20}$ while if it is one $0$ we set $a=1_{02}$. It is easy to see that the thee cases $a=1_{i_3,\dots,i_N}, 1_{20}$ and $1_{02}$ are mutually exclusive and cover all the possibilities leading to $\tilde a=1$. In summary, the output $a,b,c$ are thus given by
\be\label{eq: merging in c}
a =\begin{cases}
0 & \sum_{i=3}^N \tilde a_i = N- 3\\
1_{i_3,\dots i_N} & \sum_{i=3}^N \tilde a_i = N-2 \quad \text{and}\quad (a_3,\dots, a_N)=(1_{i_3},\dots, 1_{i_N}) \\
1_{20} & \sum_{i=3}^N \tilde a_i = N-2 \quad \text{and}\quad  {\displaystyle \min_{\substack{ 3\leq i\leq N \\ a_i=2 }} i < \min_{\substack{ 3\leq i\leq N \\ a_i=0 }} i  }\\
1_{02} & \sum_{i=3}^N \tilde a_i = N-2 \quad \text{and} \quad {\displaystyle \min_{\substack{ 3\leq i\leq N \\ a_i=2 }} i > \min_{\substack{ 3\leq i\leq N \\ a_i=0 }} i  } \\
2 & \sum_{i=3}^N \tilde a_i = N-1
\end{cases},
\qquad b=a_1 \quad \text{and}\quad c=a_2.
\ee

This induces the following distribution $P_Q^\triangle(a,b,c)$ on the triangle network, which becomes token counting upon coarse-graining $\{ 1_0, 1_1\} \mapsto 1$ to $P_Q^\triangle(\tilde a,\tilde b,\tilde c)$.

\begin{align}
 P_Q^\triangle(2,1_j, 0)&=  P_Q^\triangle(0,2, 1_j)=\frac{1}{8} u_j^2  ,  \qquad \qquad  \qquad \qquad P_Q^\triangle(0,1_j,2)= P_Q^\triangle(2,0,1_j)= \frac{1}{8} v_j^2  
 \\
P_Q^\triangle(1_{ i} ,0, 2 ) &= \begin{cases}
\frac{1}{2^N} \prod_{l=3}^{N} u_{i_l}^2 & i=(i_3,\dots i_N)\\
\frac{1}{8}- \frac{1}{2^N} & i =02 \\
0 & i =20
\end{cases},
\quad
P_Q^\triangle(1_{ i}, 2, 0 ) = \begin{cases}
\frac{1}{2^N} \prod_{l=3}^{N} v_{i_l}^2 & i=(i_3,\dots i_N)\\
0 & i =02 \\
\frac{1}{8}- \frac{1}{2^N} & i =20
\end{cases}\\
&P_Q^\triangle(1_{i}, 1_j, 1_k) = \begin{cases}
\frac{1}{2^N}\Big ( \prod_{n=1}^{N}  u_{i_n} + \prod_{n=1}^{N}  v_{i_n} \Big) ^2 &  i=(i_3,\dots i_N), \ \text{with $i_1=j, \ i_2 = k$}\\
\left(\frac{1}{8}- \frac{1}{2^N}\right) u_j^2 u_k^2 &  \ i=02\\
\left(\frac{1}{8}- \frac{1}{2^N}\right) v_j^2 v_k^2 &  \ i=20
\end{cases}.
\end{align}

To prove its nonlocality we can thus demonstrate the infeasibility of the LP in Eq.~\eqref{th3: c1}.
For simplicity, we will however relax the feasibility problem by only looking at the variables corresponding to $a=1_{\bm i}$ with $\bm i = (i_3,\dots i_N)$, and discarding those with $a=1_{20},1_{02}$. 


Assuming the existence of a trilocal model reproducing $P_Q^\triangle(a,b,c)$, from the rigidity of token counting, there must exist the variables that show the movement of each individual tokens. Considering the case when all of the outputs $\tilde a,\tilde b, \tilde c = 1$, the three tokens must be transmitted either clockwise ($\cw$) or anti-clockwise ($\acw$). With this in mind  for $\bm t \in \{ \cw, \acw\}$ let us define the following variables
\begin{align}
q(\bm i, j,k| \bm t) &=  \text{Pr} (a= 1_{\bm i},  b= 1_j , c= 1_k| \bm t )
=  \frac{\text{Pr} (a= 1_{\bm i}, b= 1_j , c= 1_k, \bm t)}{\text{Pr} (\bm t)}.
\end{align}
Here, we know that for our distribution of interest, TC rigidity implies $\text{Pr}(\bm t)= \text{Pr}(\cw)=\text{Pr}(\acw)=\frac{1}{8}$. 
Note that these variables are positive but do not define probability distributions, indeed $\sum_{\bm i, j, k} q(\bm i, j,k| \bm t) <1$ since the sum does not include any events with $a=1_{02}$ or $a=1_{20}$. Nevertheless, we have the following identity
\begin{align}\label{eq: averagae}
\frac{1}{8}(q(\bm i,j,k| \cw)+q(\bm i,j,k| \acw)) &=
\text{Pr}(\cw) q(\bm i,j,k| \cw) + \text{Pr}(\acw)q(\bm i,j,k| \acw) 
\\ &= 
  P_Q^\triangle (1_{\bm i}, 1_j , 1_k) 
 =  \frac{1}{2^N}\Big ( \prod_n  u_{i_n} + \prod_n  v_{i_n} \Big) ^2
\end{align}
Next, let us look at the marginal "network" constraints. For the parties $B$ and $C$ we have 
\begin{align}\label{eq: marj2}
   q(j|\cw)& = \sum_{\bm i,k} q(\bm i,j,k|\cw)
    =\text{Pr} (b= 1_{ j}, a\notin \{1_{02},1_{20}\}| \bm t=\cw) \\
    &=\text{Pr} (b= 1_{ j}| \bm t=\cw) -\text{Pr} (b= 1_{ j}, a= 1_{02}| \bm t=\cw) -\text{Pr} (b= 1_{ j}, a= 1_{20}| \bm t=\cw)  
    \end{align}
Here, we have 
    \begin{align}
    \text{Pr} (b= 1_{ j}| \bm t=\cw) &=\text{Pr} (b= 1_{ j}| \bm t=(1,1,1))= \text{Pr} (b= 1_ j | \bm t=(1,0,1)) = 8 P_Q^\triangle(0, 1_j, 2) =   v_j^2.
\end{align}
For the other terms note that
\begin{align}
\text{Pr}(a=1_{02},b=1_j,c=1_k|\bm t=\cw)&\leq \text{Pr}(a=1_{02}| \bm t=\cw) =  \text{Pr}(a=1_{02}| \bm t=(0,1,1))= 8 P_Q^\triangle(1_{02},2,0) =0 \\
\text{Pr}(a=1_{20}, b=1_j,c=1_k|\bm t=\acw)&\leq 8 P_Q^\triangle(1_{20},0,2)=0.
\end{align}
from which we conclude that in Eq.\eqref{eq: marj2}, $\text{Pr} (b= 1_{ j}, a= 1_{02}| \bm t=\cw) =0$.
 Now we can also write
\be
\text{Pr} (b= 1_{ j}, a= 1_{20}| \bm t=\cw) = 8 \,\text{Pr} ( a= 1_{20}, b= 1_{ j},\bm t=\cw) = 8\sum_k \text{Pr} ( a= 1_{20},b= 1_{ j},c=1_k, \bm t=\cw)
\ee
and
\be
P^\triangle_Q (1_{20},1_j,1_k)=\text{Pr} ( a= 1_{20}, b= 1_{ j},c=1_k, \bm t=\cw) + \underbrace{\text{Pr} ( a= 1_{20},b= 1_{ j},c=1_k, \bm t=\acw)}_{=0}
\ee
hence 
\be
\text{Pr} (b= 1_{ j}, a= 1_{20}| \bm t=\cw) = 8\sum_k \text{Pr} ( a= 1_{20},b= 1_{ j},c=1_k, \bm t=\cw) = 8\sum_k P^\triangle_Q (1_{20},1_j,1_k).
\ee 
Finally, this gives us the expression for the marginal 
\be
q(j|\cw) = v_j^2 -  8\sum_k P^\triangle_Q (1_{20},1_j,1_k) = \frac{v_j^2}{2^{N-3}}.
\ee
Similarly we find $ q(j|\acw) =  \frac{u_j^2}{2^{N-3}}, q(k|\acw) = \frac{u_k^2}{2^{N-3}}$ and $q(k|\cw) = \frac{v_k^2}{2^{N-3}}$.  For the party $A$, we first compute
\begin{align}
     q(\bm i|\bm t =\cw)&= \sum_{j,k} \text{Pr} (a= 1_{\bm i},b=1_j,c=1_k | \bm t=\cw) = \text{Pr} (a= 1_{\bm i} | \bm t=\cw)\\
     &= \text{Pr} (a= 1_{\bm i} | \bm t=(1,1,1)) =\text{Pr} (a= 1_{\bm i} | \bm t=(0,1,1)) \\
    &=\frac{\text{Pr} (a= 1_{\bm i},  \bm t=(0,1,1))}{\text{Pr} ( \bm t=(0,1,1))} 
    =\frac{\text{Pr}(a= 1_{\bm i}, b=2,c=0)}{1/8}=\frac{1}{2^{N-3}} \prod_{n=3}^N v_{i_n}^2,\\
       q(\bm i|\bm t =\acw)  &=\frac{\text{Pr}(a= 1_{\bm i}, b=0,c=2)}{\text{Pr} ( \bm t=(0,1,1))} =    \frac{1}{2^{N-3}}\prod_{n=3}^N u_{i_n}^2,
\end{align}
using the fact that $\bm t=(0,1,1) \Longleftrightarrow (b,c)=(2,0)$ and $\bm t=(1,0,0) \Longleftrightarrow (b,c)=(0,2)$ guaranteed by TC rigidity.

We are thus left with the following feasibility problem
\begin{align}\label{eq: LP final res 3}
\max_{\substack{q(\bm i,j,k| \cw)\geq 0 \\q(\bm i,j,k| \acw)\geq 0}} \quad & 1 \\
 &q(\bm i,j,k| \cw)+q(\bm i,j,k| \acw) = \frac{1}{2^{N-3}}\Big ( \prod_n  u_{i_n} + \prod_n  v_{i_n} \Big) ^2 \\
 & q(\bm i|\bm t =\cw) =\frac{1}{2^{N-3}} \prod_{n=3}^N v_{i_n}^2,\quad q(j|\cw) =\frac{v_j^2}{2^{N-3}}, \quad q(k|\cw) =\frac{v_k^2}{2^{N-3}} \\
&q(\bm i|\bm t =\acw) = \frac{1}{2^{N-3}}\prod_{n=3}^N u_{i_n}^2, \quad q(j|\acw) =\frac{u_j^2}{2^{N-3}},\quad
q(k|\acw) =\frac{u_k^2}{2^{N-3}}.
\end{align}
We now show analytically that this linear program is unfeasible in the limit $u\to 1$, by noting that it can be cast in a form, which was considered in Appendix C of Ref. \cite{renou2022network} and proven unfeasible.



To do so we compute $\sum_{\bm i,j,k} q(\bm i,j,k| \bm t) = \sum_j q(j|\bm t)= \frac{1}{2^{N-3}}$. This allows us to define a new variable
\be
q'(i_1,\dots,i_N,\bm t) = 2^{N-3} 
\begin{cases}
\frac{1}{2} q(\bm i=(i_3,\dots, i_N),j=i_1,k=i_2|\cw) & \bm t =\cw\\
\frac{1}{2} q(\bm i=(i_3,\dots, i_N),j=i_1,k=i_2|\acw) & \bm t=\acw.
\end{cases}
\ee
which is a probability distribution satisfying $\sum_{i_1,\dots i_N,\bm t}q'(i_1,\dots,i_N,\bm t)=1$. Noting that 
\begin{align}
q'(i_1,\dots,i_N,\cw)+q'(i_1,\dots,i_N,\acw) &= \frac{2^{N-3}}{2} (q(\bm i,j,k| \cw)+q(\bm i,j,k| \acw)) = \frac{1}{2}\Big ( \prod_k  u_{i_k} + \prod_k  v_{i_k} \Big) ^2 \\
q'(i_k|\bm t=\cw) &= 2^{N-3} \sum_{i_n|n \neq k} q(\bm i|\bm t=\cw) = v_{i_k}^2 \qquad  \forall \,k\in\{3,\dots, N\}\\
q'(i_k|\bm t= \acw) &= u_{i_k}^2 \qquad \forall,k \in\{3,\dots,N\},
\end{align}
allows us to rewrite the linear program as 
\begin{align}
\max_{q'(i_1,\dots i_N,\bm t)} \quad & 1 \\
& q'(i_1,\dots i_N,\bm t) \text{ is a probability distribution}\\
 &q'(i_1,\dots,i_N) = \frac{1}{2}\Big ( \prod_{n=1}^N  u_{i_n} + \prod_{n=1}^N  v_{i_n} \Big) ^2 \\
 & q'(i_n,\bm t =\cw) =\frac{1}{2} v_{i_n}^2 \quad \forall n \in\{1,\dots,N\}\\
 & q'(i_n,\bm t =\acw) =\frac{1}{2} u_{i_n}^2 \quad \forall n \in\{1,\dots,N\}.
\end{align}
But this is exactly the LP presented in the Claim 2 in \cite{renou2022network} for the N-party ring network. In the Proposition 2 following the claim the authors of \cite{renou2022network} show that for $N=2k+1$ and $N=4k$ the LP is infeasible for $u$ close enough to 1.

\section{Topological robustness in standard Bell scenarios}

\label{app: standard Bell}
Consider a standard bipartite Bell test, where a source $S_3$ is connected to two parties, say Alice and Bob. The parties also have access to local sources of randomness $S_1$ and $S_2$, which they use to sample the ``inputs'' labeled $x$ for Alice and $y$ for Bob. In turn, depending on those values they perform different measurements on the system received from the source $S_3$ to produce the respective ``outputs'' $a$ and $b$.
The network structure underlying this scenario is represented in Fig.~\ref{fig:standardBell}(a).  In the case of standard Bell nonlocality one usually expresses the correlation through the conditional distribution 
$P_Q^\sqcup(a,b|x,y)$ of the outputs $a$ and $b$, which is nonlocal if and only if it violates a Bell test.  However, it can be equivalently discussed through the joint distribution (as long as $P_Q^\sqcup(x)$ and $P_Q^\sqcup(y)$ are non-degenerate) 
\begin{equation}
    P_Q^\sqcup(a,b,x,y)= P_Q^\sqcup(a,b|x,y) P_Q^\sqcup(x,y) =  P_Q^\sqcup(a,b|x,y) P_Q^\sqcup(x)P_Q^\sqcup(y),
\end{equation}
where in the last equality we used the independence of $x$ and $y$ implied by the network structure. Note that in this network the sources $S_1$ and $S_2$ can be assumed classical without loss of generality.  Indeed, the measurements of $x$ and $y$ can always be performed on the systems going upwards \emph{at the source}, the systems sent down are then prepared in some states conditional on the values $x$ and $y$. However, equivalently the sources can simply send the values $x$ and $y$ downwards, while the state preparation is postponed until these systems meet with the systems prepared by $S_3$. This is true in quantum and any generalized probability theory.

Let us now assume that $P_Q^\sqcup(a,b,x,y)$ is nonlocal, and discuss its topological robustness. Note that this question is different from the analysis of nonlocality under partial independence of source $S_1$ and $S_2$ from $S_3$ ~\cite{Puetz2014}. Here we do not need to introduce any quantitative notion of independence, we only consider the bare network structure. In particular, if two sources are merged into a single one, it can correlate the underlying variables \emph{arbitrarily}. In the following table we discuss the topological robustness of standard Bell nonlocality, by looking on different deformations of the $\sqcup$-network into stronger ones, depicted in Fig.~\ref{fig:standardBell}(b)-(d).

\begin{figure}[H]
    \centering
    \includegraphics[width=\textwidth]{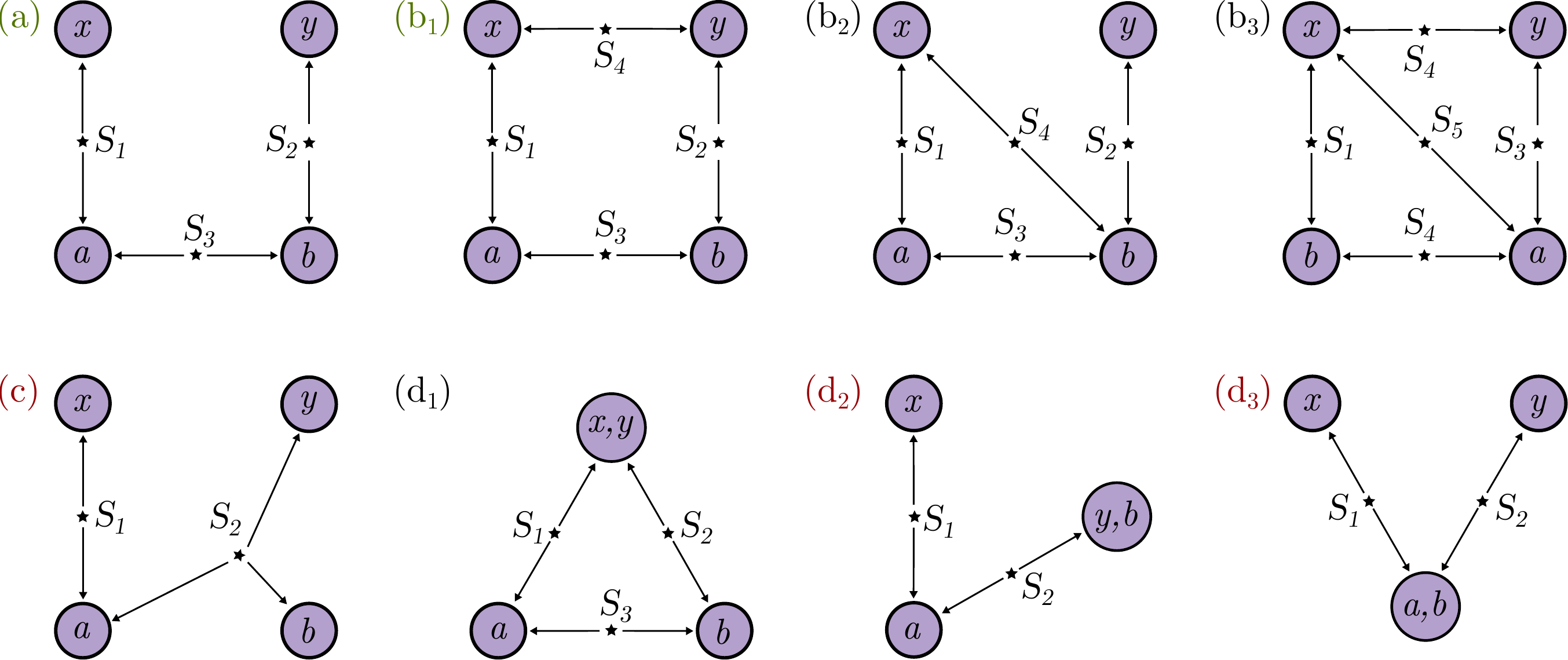}
    \caption{ \textbf{(a)} The network underlying the standard bipartite tests of Bell nonlocality exhibited by the distribution $ P_Q^\sqcup(a,b|x,y)\simeq P_Q^\sqcup(a,b,x,y)$. \textbf{(b-d)} Deformations of the original network to stronger networks. Nonlocality of $P_Q^\sqcup(a,b,x,y)$ with respect to these networks is discussed in the Table~\ref{tab:Bell}.}
    \label{fig:standardBell}
\end{figure}
\begin{table}[H]
\centering
\begin{tabular}{|c | c |l |}
\hline
Network & Nonlocality of $P_Q^\sqcup(a,b,x,y)$ & Proof or comment  \\
\hline \hline
(a) &  nonlocal \cellcolor{lime}& By assumption. $P_Q^\sqcup(a,b|x,y)$ violates a Bell inequality.
\\
\hline
 &   \cellcolor{lime}& By assumption $x$ and $y$ are sampled from $p(x,y)=p(x)p(y)$. Furthermore, they are 
\\
(b$_1$)&nonlocal\cellcolor{lime}&
manifestly independent of $b$ and $a$ respectively. Hence one can still interpret \\
&\cellcolor{lime}& $P_Q^\sqcup(a,b|x,y)$ as arising from the usual Bell test with random local settings.
\\
\hline
(b$_2$), (b$_3$) &  unknown \cellcolor{lightgray} & 
\\
\hline
(d$_1$) &   unknown \cellcolor{lightgray} & \\
\hline
(c) &  local \cellcolor{pink}& The bottom left party samples $P_Q^\sqcup(a|x,b,y)$.  Note that $x$ with $(b,y)$ are independent\\
& \cellcolor{pink}  &   and can be sampled locally by $S_1$ and $S_2$ respectively. 
\\
\hline
(d$_2$) &  local  \cellcolor{pink}& Similar to (c).\\ 
\hline
(d$_3$) &  local  \cellcolor{pink}& The bottom party samples $P_Q^\sqcup(a,b|x,y)$. Note that $x$ with $y$ are independent\\ 
 &   \cellcolor{pink}& and can be sampled by $S_1$ and $S_2$ respectively.\\ 
\hline
\end{tabular}
\caption{ Nonlocality of distributions $P_Q^\sqcup(a,b,x,y) \simeq P_Q^\sqcup(a,b|x,y)$ (standard Bell scenario) with respect to stronger networks in the Fig.~\ref{fig:standardBell}.\label{tab:Bell} }
\end{table}

\end{document}